\documentclass[12pt]{article}
\usepackage{amsmath,mathrsfs, eufrak, mathtools,amsthm}
\usepackage{amssymb}

\usepackage{bbm}
\usepackage{nicefrac}
\usepackage{newtxtext}
\usepackage{graphics,framed, wrapfig}
\usepackage{pgf}
\usepackage{tikz,pgfplots}
\usepackage{mhchem}
\usepackage{mathtools}

\usepackage[noadjust]{cite}

\usepackage[mathcal]{euscript}
\usepackage{microtype}
\usepackage{titlesec}
\usepackage[top=1.0in,bottom=1.0in,left=1.5in,right=1in]{geometry}

\titleformat{\subsubsection}[runin]{\normalfont\bfseries}{\thesubsubsection}{1em}{}

\usepackage{pstricks,pst-plot,psfrag,mathdots}
\usepackage{graphicx,subfigure,xspace,bm,tikz,pgfplots}
\usepackage[lined,linesnumbered,algoruled,noend]{algorithm2e}

\newtheorem{theorem}{Theorem}[section]
\newtheorem{definition}{Definition}

\newtheorem{lemma}{Lemma}

\newtheorem{remark}{Remark}

\newtheorem{example}{Example}

\newtheorem{proposition}[theorem]{Proposition}  
\newcounter{para}[section]

\newcommand{\oprocendsymbol}{\hbox{$\bullet$}}
\newcommand{\oprocend}{\relax\ifmmode\else\unskip\hfill\fi\oprocendsymbol}

\newcommand{\K}{\mathcal{K}}

\newcommand{\GL}{\operatorname{GL}}

\newcommand{\eig}{\operatorname{eig}}

\newcommand{\R}{\mathbb{R}}

\newcommand{\V}{\mathbb{V}}
\renewcommand{\cal}{\mathcal}

\newcommand{\dep}{\operatorname{dep}}

\definecolor{BBlue}{cmyk}{.98,0.10,0,.25}

\begin{document}

\title{Sparse Linear Ensemble Systems and \\ Structural Controllability}
\date{}
\maketitle

\vspace{-2cm}
\begin{flushright}
{\small {Xudong Chen\footnote[1]{X. Chen is with the ECEE Dept., CU Boulder.  Email: \texttt{xudong.chen@colorado.edu}.} 
}}
\end{flushright}

\begin{abstract} 
The paper introduces and solves a structural controllability problem for continuum ensembles of linear time-invariant systems. 
All the individual linear systems of an ensemble are sparse, governed by the same sparsity pattern. 
Controllability of an ensemble system is, by convention, the capability of using a common control input to simultaneously steer every individual systems in it.     
A sparsity pattern is structurally controllable if it admits a controllable linear ensemble system. A main contribution of the paper is to provide a graphical condition that is necessary and sufficient for a sparsity pattern to be structurally controllable. Like other structural problems, the property of being structural controllable is monotone. We provide a complete characterization of minimal sparsity patterns as well.       
\end{abstract}

\section{Introduction}\label{sec:introduction}
In the paper,  we introduce and solve a structural controllability problem for continuum ensembles of linear time-invariant systems. A brief description of the problem is given below. Motivations for studying the problem are given after. 

Let $\Sigma$ be a closed interval in $\R$. Consider a linear ensemble system parameterized by a variable $\sigma\in \Sigma$ as follows:
\begin{equation}\label{eq:begin}
	\dot x(t, \sigma) := \frac{\partial}{\partial t} x(t, \sigma) = A(\sigma) x(t,\sigma) + B(\sigma) u(t), 
\end{equation}
where $A:\Sigma\to \R^{n\times n}$ and $B:\Sigma\to \R^{n\times m}$ are matrix-valued functions on the interval~$\Sigma$, $x(t, \sigma)\in \R^n$ is the state of the individual system indexed by $\sigma$ at time $t$, and $u(t)\in \R^m$ is a common control input that applies to all individual systems. The $(A, B)$ pair considered here is compliant with a certain sparsity pattern, i.e., certain entries of $A$ and $B$ are zero functions. The interval $\Sigma$ is commonly referred to as the {\em parameterization space}.  
Controllability of the linear ensemble system~\eqref{eq:begin} is, roughly speaking, the capability of using the common control input $u(t)$ to simultaneously steer every individual linear control system in it. Instead of investigating controllability of a particular sparse pair $(A, B)$, we characterize sparsity patterns that admit controllable pairs. A precise problem formulation will be given in Section~\ref{sec:formulationandresult}.

The above ensemble control problem has connections with the problem of controlling a large population of recurring small networks in a much larger complex system. Recurring patterns  with significantly high frequencies of appearances in a large-scale complex system are known as motifs~\cite{milo2002network} and they are ubiquitous in nature. In many cases, steering of such complex system is often achieved by ``broadcasting'' control inputs to manipulate the network motifs.   
Notable examples include social networks where families or companies are influenced by advertisements or government policies, biological networks where gene regulatory motifs respond to external stimuli, and quantum ensembles where coupled nuclear spins are manipulated by radio-frequency pulses. 

The importance of the structures of motifs is in the belief that these structures are essential for certain functions to be achieved. The function of our interest in this paper is a fundamental one in control theory, namely, controllability.

When it comes to engineering, the framework of controlling an ensemble of relatively small-sized networks complements existing methods for controlling large-scale multi-agent systems. Many existing methods rely on the use of leader-follower hierarchies~\cite{baillieul2003information,rahmani2009controllability,pasqualetti2014controllability,chen2019controlling}; specifically, the controller steers the network by controlling only a few leading agents and, meanwhile, let the followers obey certain local feedback control laws.  
However, a larger networked system tends to be more fragile and less scalable; indeed, attacks to the leading agents or failures in critical communication links can prevent the entire system from being controllable. The ensemble control framework~\eqref{eq:begin} provides an alternative: Instead of controlling a large complex network, one can control a large population of small ones. Thus, the framework is by nature resilient: Malfunctions of nodes or links affect only the corresponding individual systems without touching the others. 

It is also worth noting that having the individual systems to be networks rather than  single agents is, in fact, critical for controllability of an ensemble system. It is well known that controllability of an individual dynamical system is far from being sufficient for an ensemble of such systems to be controllable. This is true regardless of parameterization. 
For example, an ensemble of single integrators $\dot x(t, \sigma) = B(\sigma)u(t)$, for $\sigma\in \Sigma$, can never be controllable regardless of any choice of~$B$.  
To make an ensemble system controllable, a much more stringent condition has to be met by every individual system (e.g., the $A$-matrix cannot be nilpotent as we will see later). 
However, the dynamics of individual agents often do not satisfy these conditions as was illustrated in the above example.    
A solution provided by~\eqref{eq:begin} is to let the agents form  relatively small and cooperative networks---cooperative in a sense that the connections between different agents work together to ``enrich'' the dynamics of the individual networked systems so that the necessary and/or sufficient conditions are met for ensemble controllability. 
From this perspective, the structural controllability problem we address can be viewed as a problem for characterizing the types of structures for information flows between agents that are essential to ensemble controllability.

We further note that the ensemble framework~\eqref{eq:begin} for controlling multi-agent systems is inherently scalable. The scalability is achieved by the formulation that an infinite number (continuum ensemble) of individual systems are considered. In particular, these individual systems are required to be simultaneously controllable under the same control input.  
To see why the formulation promotes scalability, we first note a simple but critical fact~\cite{dirr2021uniform,chen2020controllability}: If an ensemble system~\eqref{eq:begin} is controllable, then so is any subensemble of it---a subensemble is obtained by collecting individual systems of~\eqref{eq:begin} whose indices $\sigma$ belong to a certain closed subset $\Sigma'$ of $\Sigma$. We will review the fact at the end of Section~\ref{sec:preliminaries}. 
Different closed subsets of $\Sigma$ then correspond to different subensembles. 
Now, if we let $\Sigma' := \{\sigma_1,\ldots, \sigma_N\}$ be a finite subset of the interval $\Sigma$, then the corresponding subensemble is nothing but a finite multi-agent system. Thus, controllability of the original ensemble system~\eqref{eq:begin} guarantees controllability of the finite multi-agent system.   
We shall note that having $\Sigma$ to be an infinite set is not only sufficient for finite subensembles of~\eqref{eq:begin} to be controllable, but also necessary. 
Indeed, if  every finite $\Sigma'$ can be embedded, as a subset, into  $\Sigma$, then $\Sigma$ is necessarily infinite. In short, addressing the case where $\Sigma$ is infinite covers all finite cases.    
Scalability of the ensemble control framework then follows as a consequence:  Because a multi-agent system is treated as a finite sub-ensemble of~\eqref{eq:begin}, adding (or removing) any finite number of individual systems into (or out of) the subensemble gives rise to another subensemble. Controllability of any subensemble is guaranteed by the controllability of~\eqref{eq:begin}. 

The problem of structural controllability for linear ensemble systems is new. To the best of author's knowledge, there has not been any work in the area.  
However, the same problem for finite-dimensional linear systems was initiated by Lin almost half a century ago. 
In his seminal paper~\cite{lin1974structural}, Lin addressed the single-input case and provided a necessary and sufficient condition (using matrix forms) for sparsity patterns to be structurally controllable. 
The result was soon generalized to a multi-input case by Shields and Pearson~\cite{shields1976structural} and by Glover and Silverman~\cite{glover1976characterization}. 
For variations of the problem, we mention strong structural controllability~\cite{mayeda1979strong,chapman2013strong,jia2020unifying}, minimal controllability~\cite{liu2011controllability,commault2015single,olshevsky2014minimal,olshevsky2015minimum}, structural controllability over finite fields~\cite{sundaram2012structural}, and structural controllability for driftless bilinear control systems~\cite{tsopelakos2018classification}.

A main contribution of the paper is to provide a graphical condition that is necessary and sufficient for a sparsity pattern to be structurally controllable for linear ensemble systems. The result is formulated in  Theorem~\ref{thm:main1}. 
We compare this condition with the one for structural controllability of {\em single} linear systems: We show that the condition for linear ensemble systems is strictly stronger and provide a simple example for illustration.   
Furthermore, because the property of being structurally controllable is monotone (in a sense that if a sparsity pattern has less zero entries, then it is more likely to be structurally controllable), we also characterize sparsity patterns that are {\em minimally} structurally controllable. A precise definition will be given in Section~\ref{ssec:minimal} and the corresponding result is formulated in Theorem~\ref{thm:minimal}.

The remainder of the paper is organized as follows: In Section~\ref{sec:preliminaries}, we introduce common notations, basic notions from graph theory, and preliminaries for linear ensemble systems. In Section~\ref{sec:formulationandresult}, we formulate  the structural controllability problem and present the main results. Analysis and proofs of the results are provided in Section~\ref{sec:analysis}. We provide conclusions and outlooks in Section~\ref{sec:conclusions}. The paper has an Appendix which provides slight extensions of the main results.

\section{Preliminaries}\label{sec:preliminaries}
In the section, we gather a few common notations and present preliminaries about graph theory and control theory for linear ensemble systems.     

\vspace{.1cm}
\noindent
{\bf Notations.}
For a vector $v = (v_1,\ldots, v_n)\in \R^n$, we let $\|v\|$ be the standard Euclidean norm. We use $\operatorname{diag}(v)$ to denote a diagonal matrix, with $v_i$ the $ii$th entry.    
 
For matrices $A\in \R^{n\times n}$ and $B\in \R^{n\times m}$, we let $C(A, B)$ be the controllability matrix $C(A, B):= [B, AB, \cdots, A^{n - 1}B]$. 

Let $\Sigma$ be a closed interval in $\R$ and $M$ be a Euclidean space or a subset of it.  We denote by $\mathrm{C}^0(\Sigma, M)$ the set of continuous functions from $\Sigma$ to $M$. For any function $f\in \mathrm{C}^0(\Sigma, M)$, we let $\|f\|_{\mathrm{L}^\infty}$ be its $\mathrm{L}^\infty$-norm.    

Let $\GL(n,\R)$ be the general linear group of degree $n$, i.e., it is the set of $n\times n$ invertible matrices. If $P\in \mathrm{C}^0(\Sigma, \GL(n,\R))$, then $P^{-1}$ exists and belongs to $\mathrm{C}^0(\Sigma, \GL(n,\R))$ as well.  

Let $(A, B)$ be an element in $\mathrm{C}^0(\Sigma, \R^{n\times n}\times \R^{n\times m})$, i.e., $A$ and $B$ are continuous, matrix-valued functions. For convenience, but with slight abuse of terminology, we will still call $A$ and $B$ ``matrices'' if there is no confusion.

Let $(A', B')$ be another element in $\mathrm{C}^0(\Sigma, \R^{n\times n}\times \R^{n\times m})$. We say that $(A, B)$ and $(A',B')$ are related by a similarity transformation if there exists a $P\in \mathrm{C}^0(\Sigma, \GL(n, \R))$ such that $A' = PAP^{-1}$ and $B' = PB$.

\subsection{Basic Notions from Graph Theory}\label{ssec:graphtheory} 
Let $G = (V, E)$ be a directed graph (or digraph), with $V = \{v_1,\ldots, v_n\}$ the node set and $E$ the edge set. 
We allow $G$ to have self-arcs. A digraph without self-arcs will be referred to as a simple digraph.  

An edge from $v_i$ to $v_j$ is denoted by $v_{i}v_j$. We call $v_j$ an out-neighbor of $v_i$ and $v_i$ an in-neighbor of $v_j$. For a given subset $V'$ of $V$, we let $N_{\rm in}(V')$ be the set of in-neighbors of $V'$ within $G$, i.e., a node $v_i$ belongs to
$N_{\rm in}(V')$ if there exist a node $v_j$ in $V'$ and an edge $v_iv_j$ in $G$. In case we need to emphasize the role of the digraph $G$, we will write $N_{\rm in}(V'; G)$. Similarly, we let $N_{\rm out}(V')$ (or $N_{\rm out}(V';G)$) be the set of out-neighbors of $V'$ in $G$.

A walk from $v_i$ to $v_j$ is a sequence of nodes $v_{i_1}\ldots v_{i_k}$, with $v_{i_1} = v_i$ and $v_{i_k} = v_j$, such that each $v_{i_j}v_{i_{j + 1}}$, for $j = 1,\ldots,k - 1$, is an edge of $G$.
The length of the walk is the number of edges contained in it. A walk is a path if there is no repetition of nodes in the sequence. A walk is a cycle if there is no repetition of nodes except the repetition of starting- and ending-nodes. Note that a self-arc at a node $v_i$ is a cycle of length~$1$.

A digraph $G$ is strongly connected if for any two different nodes $v_i$ and $v_j$, there is a path from $v_i$ to $v_j$. In particular, if $G$ is a digraph with only a single node (with or without a self-arc), then $G$ is strongly connected.  
A digraph $G$ is rooted if there is a node $v_0$ such that for any other $v_i$ in $G$, there is a path from $v_0$ to $v_i$. The node $v_0$ is a root of $G$.

A subgraph $G' = (V', E')$ of $G$ satisfies $V'\subseteq V$ and $E'\subseteq E$. Given a subset $V'$ of $V$, a subgraph $G' = (V', E')$ is said to be induced by $V'$ if the edge set $E'$ satisfies the following condition: For any two nodes $v_i$ and $v_j$ in $V'$,  $v_iv_j$ is an edge of $G'$ if and only if it is an edge of $G$.

 We say that two subgraphs $G'$ and $G''$ are disjoint if their node sets are disjoint. For a collection of pair-wise disjoint subgraphs $G_1 = (V_1, E_1),\ldots, G_k = (V_k, E_k)$ of $G$, we let their disjoint union be defined as a digraph $G' = (V', E')$ with $V' := \sqcup^k_{i = 1} V_i$ and $E':=\sqcup^k_{i = 1} E_i$. 

A digraph $G$ is acyclic if it does not contain any cycle as its subgraph. If $G$ is also rooted, then it has a unique root. A directed tree (also known as an arborescence) is a special rooted acyclic digraph such that every node, except the root node $v_0$, has only one in-neighbor. It follows that for any given node $v_i$ other than $v_0$, there is a unique path from $v_0$ to $v_i$. The depth of the node $v_i$ is the length of the path. The depth of the root $v_0$ is $0$ by default. The depth of the tree $G$ is  the maximal value of depths of all the nodes.  

Note that if $G = (V, E)$ is rooted with $v_0$ a root, then it contains a tree $G' = (V, E')$, with the same node set, as a subgraph such that $v_0$ is the root of $G'$. The subgraph $G'$ is called a directed spanning tree of $G$.

\subsection{Control Theory for Linear Ensemble Systems}\label{ssec:linearensemblesystem}
Let $\Sigma$ be the unit closed interval $\Sigma := [0,1]$ in $\R$. The choice of the closed interval is for ease of presentation. The results established in the paper do not depend on a particular choice of interval as we will see later in Prop.~\ref{prop:union}.  
For convenience, we reproduce below the linear ensemble system~\eqref{eq:begin}: 
\begin{equation}\label{eq:linearsys}
	\dot x(t, \sigma) = A(\sigma) x(t,\sigma) + B(\sigma) u(t), \quad \forall \sigma\in \Sigma, 
\end{equation}
where $A:\Sigma \to \R^{n\times n}$ and $B: \Sigma\to \R^{n\times m}$ are continuous functions.  
The control input $u$ is said to be {\em admissible} if for any given time interval $[0,T]$, the function $u: [0, T] \to \R^m$ is integrable.

Let $x_\Sigma(t): \Sigma \to \R^n$ be the map that sends $\sigma$ to $x(t, \sigma)$.   
	We call $x_\Sigma(t)$ a {\em profile} at time~$t$.  In the paper, we consider only continuous profiles, i.e.,  $x_\Sigma(t)\in \mathrm{C}^0(\Sigma, \R^n)$.

 We now have the following definition:
	
\begin{definition}\label{def:controllability}
 	The linear ensemble system~\eqref{eq:linearsys} is \textbf{uniformly controllable} if for any initial profile $x_\Sigma(0)\in {\rm C}^0(\Sigma, \R^n)$, any target profile $\hat x_\Sigma\in {\rm C}^0(\Sigma, \R^n)$, and any error tolerance $\epsilon > 0$, there exist a time\footnote{It is known~\cite{triggiani1975controllability} that if system~\eqref{eq:linearsys} is uniformly controllable for some $T$, then it is uniformly controllable for {\em all} $T > 0$.} $T > 0$ and an admissible control input $u:[0,T]\to \R^m$ such that the solution $x_\Sigma(t)$  generated by~\eqref{eq:linearsys} satisfies $\|x_\Sigma(T) - \hat x_\Sigma\|_{\rm L^\infty} < \epsilon$.    
 \end{definition}

Because system~\eqref{eq:linearsys} is completely determined by the $(A, B)$ pair, we will some time use the pair to denote the system and simply say that $(A, B)$ is uniformly controllable.

Necessary and/or sufficient conditions for uniform controllability of system~\eqref{eq:linearsys} have widely been investigated in the literature (see, for example,~\cite{li2007ensemble,helmke2014uniform,li2015ensemble,dirr2021uniform}). We present below a condition that utilizes the notion of controllable subspace. For that, we first have the following definition:   	 
	 
\begin{definition}
	Let the $(A,B)$ pair be given in~\eqref{eq:linearsys}.  
	Let $\mathcal{L}(A, B)$ be the ${\rm L}^\infty$-closure of the vector space spanned by the columns of $A^kB$, for all $k \ge 0$.
	We call $\mathcal{L}(A, B)$ the \textbf{controllable subspace} associated with system~\eqref{eq:linearsys}.  
\end{definition}

The following necessary and sufficient condition, adapted from~\cite{triggiani1975controllability}, is a straightforward generalization of the Kalman rank condition for finite-dimensional linear systems:

\begin{lemma}\label{lem:controllablesubspace}
	A pair $(A, B)\in \mathrm{C}^0(\Sigma, \R^{n\times n}\times \R^{n\times m})$ is uniformly controllable if and only if 
	$\mathcal{L}(A, B) = \mathrm{C}^0(\Sigma, \R^n)$. 
\end{lemma}

We now relate controllability of system~\eqref{eq:linearsys} to controllability of its subensembles. 
 Specifically, let $\Sigma'$ be a closed subset of $\Sigma$ and consider the following ensemble system:
 \begin{equation}\label{eq:subensemble}
 \dot x(t,\sigma) = A(\sigma)x(t,\sigma) + B(\sigma) u(t), \quad \forall\sigma\in \Sigma', 
 \end{equation}
where the $(A, B)$ pair is the same as the one for system~\eqref{eq:linearsys}, but with their domains restricted to $\Sigma'$. We call system~\eqref{eq:subensemble} the \textit{\textbf{subensemble-$\Sigma'$}} of system~\eqref{eq:linearsys}. The following result is known (see, for example,~\cite{dirr2021uniform,chen2020controllability}):

\begin{lemma}\label{lem:subensemble}
	If system~\eqref{eq:linearsys} is uniformly controllable, then so is system~\eqref{eq:subensemble}.
\end{lemma}

Finally, note that there are other controllability notions associated with system~\eqref{eq:linearsys}, such as $\mathrm{L}^p$-controllability, for $p \geq 1$. We review these notions and the corresponding controllability results in the Appendix.

\section{Problem Formulation and Main Results}\label{sec:formulationandresult}
In this section, we formulate the structural controllability problem for linear ensemble systems and provide complete solutions.   
We introduce key definitions and the problem in Section~\ref{ssec:sparse}. 
Then, in  Section~\ref{ssec:main1}, we present a necessary and sufficient condition for a sparsity pattern to be structurally controllable (the result is formulated in Theorem~\ref{thm:main1}).     
Finally, in Section~\ref{ssec:minimal}, we focus on  sparsity patterns with minimal numbers of nonzero entries. A complete characterization of these patterns is provided in Theorem~\ref{thm:minimal}.    

\subsection{Problem Formulation for Structural Controllability}\label{ssec:sparse}
We still let $\Sigma$ be the unit closed interval $[0,1]$ and note again that the choice of the interval is irrelevant (see Prop.~\ref{prop:union} in the subsection). 
Let $(A, B)\in \mathrm{C}^0(\Sigma, \R^{n\times n}\times \R^{n\times m})$ be a sparse matrix pair.   
By convention, we will use a digraph $G$ to describe the sparsity pattern of $(A,B)$. The construction of the digraph is given in the following definition: 

\begin{definition}
For a pair $(A, B)\in \mathrm{C}^0(\Sigma, \R^{n\times n}\times \R^{n\times m})$, we define a digraph $G = (V, E)$ on $(n + m)$ nodes as follows: The node set $V$ is a disjoint union of two subsets $V_\alpha := \{\alpha_1,\ldots, \alpha_{n}\}$ and $V_\beta := \{\beta_1,\ldots, \beta_m\}$.    
The edge set $E$ is determined by the following two items:  
\begin{enumerate}
\item There is an edge from $\alpha_j$ to $\alpha_i$ if the $ij$th entry of $A$ is not the zero function.
\item There is an edge from $\beta_j$ to $\alpha_i$ if the $ij$th entry of $B$ is not the zero function. 
\end{enumerate}  
We call $G$ the \textbf{\textit{digraph induced by matrix pair}} $(A, B)$. The nodes in $V_\alpha$ and the nodes in $V_\beta$ are termed \textbf{\textit{state-nodes}} and \textbf{\textit{control-nodes}}, respectively. 
\end{definition}

Note that the control-nodes of $G$ do not have any incoming neighbor.  
For given nonnegative integers $n$ and $m$, let $\mathcal{G}_{n,m}$ be the set of digraphs $G$ induced by matrix pairs $(A, B)\in \mathrm{C}^0(\Sigma, \R^{n\times n}\times \R^{n\times m})$. Equivalently, a digraph $G$ belongs to $\mathcal{G}_{n,m}$ if it has $(n + m)$ nodes $\alpha_1,\ldots, \alpha_{n}$ and $\beta_{1},\ldots, \beta_{m}$ and the $\beta$-nodes do not have incoming neighbors.  
For a later purpose, we allow $n$ or $m$ to be $0$. If $n = 0$ (resp. $m = 0$), then there is no state-node (resp. control-node) in $G$. We let  
$$\mathcal{G}:= \cup^\infty_{n,m = 0} \mathcal{G}_{n,m}.$$ 
Every graph $G\in \mathcal{G}$ then corresponds to a sparsity pattern.

Conversely, for any given such digraph $G$, we introduce a class of sparse pairs $(A, B)$ that correspond to it. The correspondence is given in the following definition:

\begin{definition}
A pair $(A, B)\in \mathrm{C}^0(\Sigma, \R^{n\times n}\times \R^{n\times m})$ is \textit{\textbf{compliant}} with $G = (V, E)\in \mathcal{G}_{n,m}$ if the digraph $G' = (V, E')$ induced by $(A,B)$ is a subgraph of $G$, i.e., $E'\subseteq E$. \end{definition}

Let $(A, B)$ be a pair compliant with a digraph $G\in \mathcal{G}$. An entry $a_{ij}$ of $A$ or an entry $b_{ij}$ of $B$ is said to be a $\star$-entry if $\alpha_j \alpha_i$ or $\beta_j\alpha_i$ is an edge of $G$. 
The $\star$-entries can be arbitrary continuous functions from $\Sigma$ to $\R$. The other entries of $A$ or $B$ have to be identically zero. 
See Fig.~\ref{fig:sparsepair} for an illustration.

\begin{figure}[ht]
\begin{center}
\includegraphics[width = 0.6\textwidth]{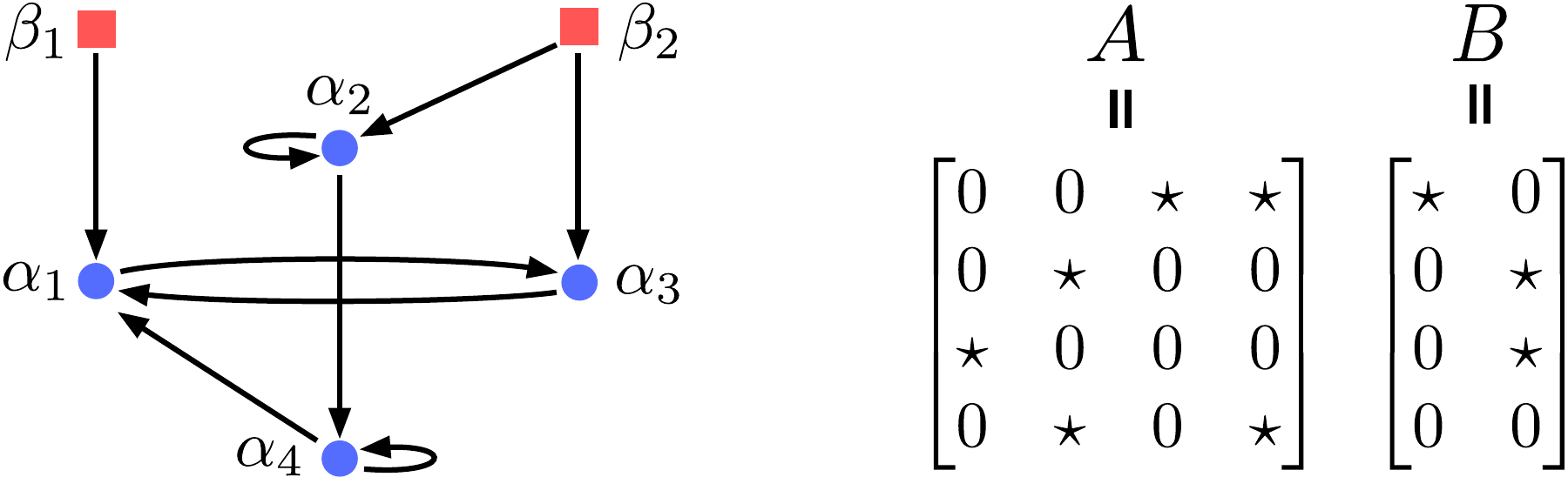}
\caption{\small {\em Left:} A digraph $G\in \mathcal{G}_{4,2}$ with four state-nodes (blue dots) and two control-nodes (red squares). {\em Right:} Sparse matrices $(A, B)$ compliant with $G$. The $\star$-entries correspond to the edges of~$G$.}\label{fig:sparsepair}
\end{center}
\end{figure}

For a given digraph $G\in\mathcal{G}_{n,m}$, we let $\mathbb{V}(G)$ be the set of matrix pairs $(A,B)$ compliant with $G$: 
$$
\mathbb{V}(G):= \{(A, B)\in \mathrm{C}^0(\Sigma, \R^{n\times n} \times \R^{n\times m}) \mid 
 (A, B) \mbox{ is compliant with } G \}.
$$  
We look for pairs $(A, B)\in \mathbb{V}(G)$ that are uniformly controllable. Structural controllability of $G$ relies on the existence of these pairs. 
Precisely, we have the following definition:

\begin{definition}\label{def:StrucContr}
	A digraph $G\in \mathcal{G}$ is \textbf{structurally controllable} if there exists a uniformly controllable pair $(A, B)\in \mathbb{V}(G)$.  
\end{definition}

If $G$ has no state-node (i.e., $n = 0$), then it is structurally controllable by default.

Apparently, the definition of structural controllability depends on the underlying parameterization space. We have so far assumed that $\Sigma$ is the closed unit interval $\Sigma = [0,1]$. 
The following fact will relax the constraint and establishes equivalence of structural controllability for a class of parameterization spaces:

\begin{proposition}\label{prop:union}
A digraph $G\in \mathcal{G}$ is structurally controllable for $\Sigma = [0,1]$ if and only if it is  structurally controllable for any finite union of closed intervals in $\R$. 
\end{proposition}

\begin{proof}
	We first show that $G$ is structurally controllable for $\Sigma$ if and only if it is for an arbitrary closed interval $\Sigma' := [r_1, r_2]$, with $r_1 < r_2$. For a given continuous pair $(A, B)$ on $\Sigma$, we define a continuous pair $(A',B')$ on $\Sigma'$ as follows: 
	\begin{equation*}\label{eq:defA'B'}
	A'(\sigma') := A\left (\nicefrac{\sigma' - r_1}{r_2 - r_1} \right ) \mbox{ and } 
	B'(\sigma') := B\left (\nicefrac{\sigma' - r_1}{r_2 - r_1} \right ).
	\end{equation*} 
	It should be clear that if a function $f$ belongs to the controllable subspace $\mathcal{L}(A, B)$, then the function $f':\Sigma'\to\R$ defined by 
	$f'(\sigma'):= f(\nicefrac{\sigma' - r_1}{r_2 - r_1})$,   
	belongs to the controllable subspace $\mathcal{L}(A', B')$.  
	Conversely, if $f'\in \mathcal{L}(A', B')$, then the function $f:\Sigma\to \R$ defined by $f(\sigma):= f'(r_1 + (r_2 -r_1)\sigma)$ 
	belongs to $\mathcal{L}(A, B)$. 
	 Thus, the two Banach spaces $\mathcal{L}(A, B)$ and $\mathcal{L}(A',B')$ are isomorphic. By Lemma~\ref{lem:controllablesubspace}, $(A, B)$ is uniformly controllable if and only if $(A',B')$ is. 
	Moreover, the pair $(A', B')$ is compliant with $G$ if and only if $(A,B)$ is. It then follows that $G$ is structurally controllable for $\Sigma$ if and only if it is for~$\Sigma'$. 
	
	We next let $\Sigma'$ be a finite union of pairwise disjoint closed intervals, i.e., $\Sigma'= \cup_{i = 1}^k\Sigma'_i$. We show that if $G$ is structurally controllable for $\Sigma$ if and only if it is structurally controllable for $\Sigma'$.   
	First, we assume that $G$ is structurally controllable for $\Sigma'$. Then, by Lemma~\ref{lem:subensemble}, it has to be structurally controllable for every single closed interval $\Sigma'_i$. It follows from the previous arguments that $G$ is also structurally controllable for the unit closed interval $\Sigma$. 
	We now assume that $G$ is structurally controllable for $\Sigma$. Using again the above arguments, we have that $G$ is uniformly controllable for any closed interval in $\R$.  
	Let $\Sigma''$ be a closed interval, sufficiently large, such that it contains $\Sigma'$ as a subset. 
	Because $G$ is structurally controllable for $\Sigma''$, we conclude from Lemma~\ref{lem:subensemble} that $G$ is structurally controllable for $\Sigma'$. This completes the proof.   
\end{proof}

Note that if $\Sigma$ is not a finite union of closed intervals, then the class of structurally controllable digraphs can be completely different. 
A case of particular interest is that $\Sigma$ is a circle. Note that every closed interval (or a finite union of them) can be embedded into a circle, but not the other way around. Thus, the class of structurally controllable digraphs for a circle is a subset of the class for a closed interval. 
For continuum spaces whose dimensions are greater than one, we conjecture that there does not exist any structural controllable digraph. 
This conjecture is based upon a recent negative result~\cite{chen2020controllability} which says that any real-analytic linear ensemble system is not uniformly controllable if the dimension of the underlying parameterization space is greater than one.  
 
\subsection{A Necessary and Sufficient Condition}\label{ssec:main1} 
In this subsection, we provide a necessary and sufficient condition for a digraph $G\in \mathcal{G}$ to be structurally controllable. The condition comprises two parts: One is about  accessibility of $G$ to the control-nodes and the other one is about existence of Hamiltonian decomposition admitted by the state-nodes of~$G$.
We give precise definitions below.

\begin{definition}
A digraph $G\in \mathcal{G}$ is \textbf{accessible} to control-nodes if for each state-node $\alpha_j$, there exist a control-node $\beta_i$ and a path from $\beta_i$ to $\alpha_j$. 
\end{definition}

We also need the following definition:

\begin{definition}
Let $H = (V, E)$ be an arbitrary digraph. The digraph $H$ admits a \textbf{Hamiltonian decomposition} if it contains a subgraph $H' = (V, E')$, with the same node set $V$ and $E' \subseteq E$, such that $H'$ is a disjoint union of cycles.  	
\end{definition}

With the above definitions, we will now present the first main result of the paper:

\begin{theorem}\label{thm:main1}
	A digraph $G\in\mathcal{G}$ is structurally controllable if and only if the following hold:
	\begin{itemize}
	\item[A1.] The digraph $G$ is accessible to control-nodes. 
	\item[A2.] The subgraph $H$ induced by the state-nodes admits a Hamiltonian decomposition.  
	\end{itemize}
\end{theorem}
 
\begin{definition}
The two items A1 and A2 combined will be referred to as \textbf{\textit{condition-A}}.   
\end{definition}

For illustration, we consider the digraph $G$ in Fig.~\ref{fig:sparsepair}. First, note that $G$ is accessible: There are edges $\beta_1\alpha_1$, $\beta_2\alpha_2$, $\beta_2\alpha_3$, and a path $\beta_2\alpha_2\alpha_4$, ending with the four state-nodes.  
	Next, note that the subgraph $H$ induced by the state-nodes admits a Hamiltonian decomposition: Nodes $\alpha_1$ and $\alpha_3$ form a $2$-cycle and the remaining two nodes $\alpha_2$, $\alpha_4$ have self-arcs. By Theorem~\ref{thm:main1}, $G$ is structurally controllable. 

It is known (see, e.g.,~\cite[Lemma 1]{helmke2014uniform}) that if a linear ensemble system is uniformly controllable, then all of its individual systems are controllable. Thus, if $G$ is structurally controllable for linear ensemble systems, then it is structurally controllable for finite-dimensional linear systems. However, the converse is not true. We elaborate below on this fact:

\begin{remark}\label{rmk:comparison}
{\normalfont 
For finite-dimensional linear systems,  
a necessary and sufficient condition~\cite{lin1974structural,shields1976structural,olshevsky2015minimum} for structural controllability 
can be formulated as follows: A digraph $G\in \mathcal{G}$ is structural controllable if and only if it satisfies item A1 in Theorem~\ref{thm:main1} and the following item:   
\begin{enumerate}
	\item[\em C2.] {\em For any subset $V'$ of state-nodes of $G$,  $|N_{\rm in}(V')| \ge |V'|$.}
\end{enumerate}
We recall that $N_{\rm in}(V')$ is the set of in-neighbors of $V'$. 
The above item C2 is strictly weaker than the item A2 in Theorem~\ref{thm:main1}. To see this, we let $G$ satisfy A2 and $H$ be the subgraph induced by the state-nodes. Let $H'$ be a disjoint union of cycles that cover all the state-nodes. Then, within the subgraph $H'$, we have that for any subset $V'$ of state-nodes, $|N_{\rm in}(V'; H')| = V'$. It then follows that 
$$
|N_{\rm in}(V')| = |N_{\rm in}(V'; G)|  \geq |N_{\rm in}(V'; H')| = |V'|.      
$$
On the other hand, there exist digraphs that satisfy items A1 and C2, but not A2.  One can simply take the class of directed paths as an example  (see Fig.~\ref{fig:comparison}).   
\begin{figure}[ht]
\begin{center}
\includegraphics[width = 0.45\textwidth]{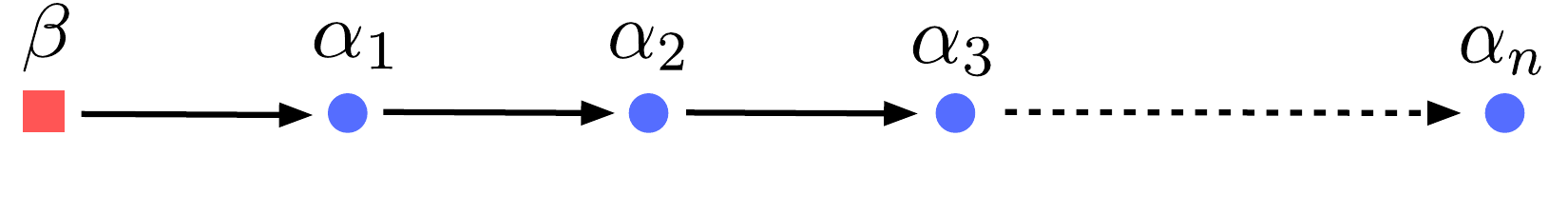}
\caption{\small Consider the path digraph $G_{n,1}$ on $n$ state-nodes and a single control-node, with edges $\beta \alpha_1$ and $\alpha_{i}\alpha_{i+1}$ for $i = 1,\ldots, n-1$. The path digraph is structural controllable for finite-dimensional linear systems; indeed, the $n$th order integrator $\frac{d^n}{dt^n} x(t) = u(t)$ is controllable and the corresponding matrix pair $(A, B)$ is compliant with $G_{n,1}$. However, by Theorem~\ref{thm:main1}, the path digraph is {\em not} structural controllable for linear ensemble systems because the subgraph $H$ induced by the $\alpha$-nodes does not admit a Hamiltonian decomposition. This issue can be resolved by adding, e.g., the edge $\alpha_n\alpha_1$ so that the resulting subgraph $H$ is a Hamiltonian cycle.}\label{fig:comparison}
\end{center}
\end{figure} 
}
\end{remark}

If $G$ has several connected components, then each component corresponds to a sparse linear ensemble system. The dynamics of these ensemble systems are decoupled from each other. It follows that $G$ is structurally controllable if and only if every connected component of $G$ satisfies condition-A.

As was mentioned earlier, there are variations on controllability notions for linear ensemble systems: One can replace uniform controllability with  $\mathrm{L}^p$-controllability, for $1\leq p < \infty$, which is known to be weaker. Correspondingly, one relaxes Def.~\ref{def:StrucContr} as follows: A digraph $G\in \mathcal{G}$ is structural controllability if there exists an $\mathrm{L}^p$-controllable pair compliant with $G$.      
With such relaxation, one may wonder whether condition-A is still necessary and sufficient? The answer is affirmative and, in fact, the proof of this result can be obtained, with slight modification, from the proof of Theorem~\ref{thm:main1} given in the next section. We elaborate on the above arguments in the Appendix.

\subsection{Characterization of Minimal Digraphs}\label{ssec:minimal}
In this subsection, we focus on a special class of structurally controllable digraphs, namely, digraphs with minimal numbers of edges. These digraphs corresponds to the sparsity patterns with minimal numbers of $\star$-entries. To that end, we have the following definition: 

\begin{definition}
	A structurally controllable digraph $G\in \mathcal{G}$ is \textbf{minimal} if removal of any edge out of $G$ causes the digraph to lose structural controllability.    
\end{definition}

We provide below a complete characterization of minimally structurally controllable digraphs. For that, we need a few preliminaries, and start with the following definition:

\begin{definition}\label{def:strongcompdecomp} 
Let $G = (V, E)$ be an arbitrary weakly connected digraph. The \textbf{strong component decomposition} $V = \sqcup^N_{i = 0} V_i$ satisfies the following conditions:  
\begin{enumerate}
\item Let $G_i$ be the subgraph of $G$ induced by $V_i$. Then, every $G_i$ is strongly connected.
\item If $G'$ is another induced subgraph of $G$ and is strongly connected, then $G'$ has to be a subgraph of $G_i$ for some $i \in \{0,\ldots, N\}$. 
\end{enumerate}
\end{definition}

We count the number of strong components from $0$ because, later, we will use $G_0$ to denote the singleton formed by the unique control-node of a digraph $G\in \mathcal{G}_{n,1}$. It will thus distinguish itself from others $G_1,\ldots, G_N$, which are formed by state-nodes. 

The strong component decomposition exists and is unique (see, for example,~\cite{chen2015controllability}). By condensing these strong components into  single nodes, one obtains a simple digraph $S$:

\begin{definition}
Let  $V = \sqcup_{i = 0}^N V_i$ be the strong component decomposition of $G$. 
The \textbf{skeleton digraph} $S$ of $G$ is defined as follows: 
There are $(N + 1)$ nodes $w_0,\ldots,w_N$ in $S$, corresponding to the $(N + 1)$ strong components. The digraph $S$ does not have self-arcs. For two different nodes $w_i$ and $w_j$, there is an edge $w_iw_j$ if and only if there is an edge $v_i v_j$ in $G$ with $v_i\in V_i$  and $v_j \in V_j$. 
\end{definition}

It should be clear that $S$ is acyclic.  
To every edge $w_i w_j$ of $S$, we define a subset of edges of $G$ as follows:
\begin{equation}\label{eq:squarewiwj}
[w_i w_j] := \{v_i v_j \in E \mid v_i\in V_i \mbox{ and } v_j \in V_j\},
\end{equation}    
i.e., $[w_iw_j]$ is the collection of edges from $G_i$ to $G_j$. 
By the construction of skeleton digraph, the set $[w_i w_j]$ is nonempty. 

We now apply condensation to the digraphs $G\in\mathcal{G}$ and obtain their skeleton digraphs~$S$. 
Note that each control-node $\beta_i$ of $G$ is itself a strongly connected component and, hence, gives rise to a node of the skeleton digraph $S$.  
Other strongly connected components $G_j$ of~$G$ are all contained in the subgraph $H$ induced by the state-nodes.   

Also, note that if $G$ has only one control-node and if $G$ is structurally controllable, then by Theorem~\ref{thm:main1}, the skeleton digraph $S$ is rooted acyclic. The unique root of $S$ corresponds to the control-node of $G$. 

We further recall that an arborescence is a directed rooted tree. 
  With the above preliminaries, we now have the second main result that characterizes all minimally structurally controllable digraphs:

\begin{theorem}\label{thm:minimal}
	A weakly connected digraph $G\in\mathcal{G}$ is minimally structurally controllable if and only if the following hold: 
	\begin{enumerate}
		\item[B1.] There is only one control-node $\beta$. The skeleton digraph $S$ of $G$ is an arborescence. Moreover, for every edge $w_iw_j$ of $S$, the set $[w_iw_j]$ defined in~\eqref{eq:squarewiwj} is a singleton. 
		\item[B2.] Let $G_0,\ldots, G_N$ be the subgraphs of $G$ obtained from the strong component decomposition, with $G_0$ the singleton $\{\beta\}$. Then, every $G_i$, for $i = 1,\ldots, N$, is a cycle.   
	\end{enumerate}
\end{theorem}
 
\begin{definition} 
The two items B1 and B2 combined will be referred to as \textbf{\textit{condition-B}}. 
\end{definition}
 
For illustration, we provide in Fig.~\ref{fig:niceexample} all weakly connected, minimally structurally controllable digraphs $G$ with three state-nodes (the number of control-nodes is necessarily one).   

\begin{figure}[ht]
\begin{center}
\includegraphics[width = 0.6\textwidth]{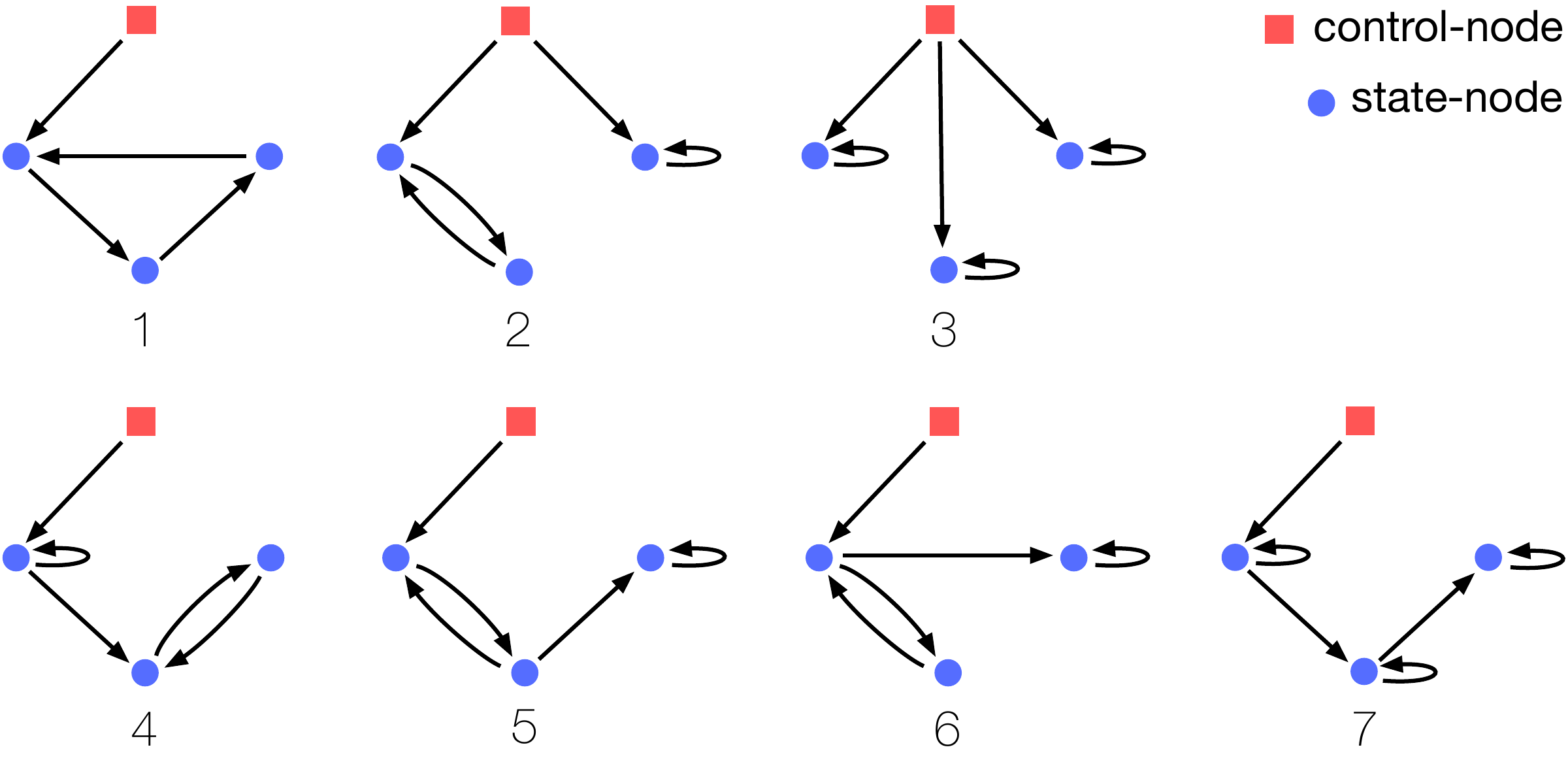}
\caption{\small We enumerate in the figure all seven weakly connected, minimally structurally controllable digraphs with three state-nodes. The number of control-nodes is necessarily one.}\label{fig:niceexample}
\end{center}
\end{figure}

\section{Analysis and Proofs of Main Results}\label{sec:analysis}
This section is devoted to the proofs of the two main results, Theorems~\ref{thm:main1} and~\ref{thm:minimal}, formulated in the previous section.    
The analysis comprises three parts: 
\begin{enumerate}
\item 	In Section~\ref{ssec:necessity}, we show that condition-A is necessary for structural controllability. This part is more or less straightforward. 
\item In Section~\ref{ssec:minimalB}, we show that condition-B is minimal with respect to condition-A, i.e., every digraph satisfying condition-A can be reduced, via edge deletion, to a disjoint union of digraphs satisfying condition-B. In the same subsection, 
we  also recall the fact that the property of being structural controllable is monotone with respect to edge-set inclusion. Thus, to establish sufficiency of condition-A, it suffices to establish sufficiency of condition-B. 
\item In Section~\ref{ssec:matrixform}, we represent minimal sparsity patterns in matrix forms. This prepares for explicit constructions of uniformly controllable pairs $(A, B)$, which will be carried out in Section~\ref{ssec:sufficiency}.   
\end{enumerate}

\subsection{Necessity of Condition-A}\label{ssec:necessity}
In this subsection, we establish the following result: 

\begin{proposition}\label{prop:necessary}
If $G$ is structurally controllable, then $G$ satisfies condition-A given in the statement of Theorem~\ref{thm:main1}. 
\end{proposition}

\begin{proof}
We need to show that $G$ is accessible to control-nodes and that the subgraph $H$ induced by the state-nodes admits a Hamiltonian decomposition. 
\vspace{.1cm}
 	
\noindent
{\em Proof that $G$ is accessible.}  
Recall that $V_\alpha$ is the set of state-nodes in $G$. Suppose, to the contrary, that $G$ is not accessible; then, we can partition the set $V_\alpha$ into two nonempty subsets: $V_\alpha = V^+_\alpha \sqcup V^-_\alpha$. The subset $V^+_\alpha$ is the collection of nodes to which there exist paths from the control-nodes and $V^-_\alpha:= V_\alpha \backslash V^+_\alpha$ is the complement of~$V^+_\alpha$ in $V_\alpha$.  
Let $k:= |V^-_\alpha|$; then, $1\le k \le n$. By relabeling the nodes, if necessary, we can assume that $V^-_\alpha$ comprises the last $k$ nodes $\alpha_{n - k +1},\ldots,\alpha_n$. 

We next pick an arbitrary pair $(A, B)\in \mathbb{V}(G)$, and partition $A$ and $B$ into blocks: $A = [A_{11}, A_{12}; A_{21}, A_{22}]$ and $B = [B_1; B_2]$, 
where $A_{11}$ is $k\times k$ and $B_1$ is $k\times m$. By construction of $V^-_\alpha$ and $V^+_\alpha$, we have that the blocks $A_{21}$ and $B_2$ are zeros. Thus, the corresponding ensemble system is in the Kalman canonical form:
\begin{equation*}\label{eq:kalmancf}
	\begin{bmatrix}
		\dot x_1(t,\sigma) \\
		\dot x_2(t,\sigma)
	\end{bmatrix} = 
	\begin{bmatrix}
 A_{11}(\sigma) & A_{12}(\sigma) \\
 0 & A_{22}(\sigma) 	
\end{bmatrix}
\begin{bmatrix}
		x_1(t,\sigma) \\
		x_2(t,\sigma)
	\end{bmatrix}  +  
	\begin{bmatrix} 
		B_1(\sigma) \\
		0 
	\end{bmatrix} u(t), \quad\forall \sigma\in \Sigma.
\end{equation*}

We claim that the pair $(A, B)$ is not uniformly controllable. 
To see this, let $f\in \mathcal{L}(A, B)$, and we decompose $f = [f_1; f_2]$ with $f_2$ of dimension~$k$. Then, $f_2 = 0$, and the claim follows from Lemma~\ref{lem:controllablesubspace}. 
\vspace{.1cm}

\noindent
{\em Proof that $H$ admits a Hamiltonian decomposition.} 
To proceed, we first recall the following necessary condition~\cite[Lemma 1]{helmke2014uniform} for a continuous matrix pair $(A,B)$ to be uniformly controllable: If $(A, B)$ is uniformly controllable and if $B$ has $m$ columns, then for any finite number $q\geq (m + 1)$ of distinct points $\sigma_1,\ldots,\sigma_{q}$ in $\Sigma$, we have that   
\begin{equation}\label{eq:intersectionofeigensets}
\eig(A(\sigma_1))\cap \cdots \cap \eig(A(\sigma_q)) = \varnothing,  
\end{equation}  
where $\eig(A(\sigma_i))$ is the set of eigenvalues of $A(\sigma_i)$. 

Next, recall that $H$ is the subgraph of $G$ induced by the state-nodes. Thus, the sparsity pattern of matrix $A$ is determined by $H$. We then make the following observation 	(adapted from~\cite{belabbas2013sparse}): 
If the graph $H$ does not admit a Hamiltonian decomposition, then for any pair $(A, B)\in \mathbb{V}(G)$, the determinant of $A$ is identically zero, i.e., 
\begin{equation*}\label{eq:determinantiszero}
\det A(\sigma) = 0, \quad \forall \sigma\in \Sigma.
\end{equation*} 

It then follows that for any $m\ge 0$ and for any $q\ge (m + 1)$ distinct points $\sigma_1,\ldots,\sigma_q$ in $\Sigma$, we have that
$$
0\in \eig(A(\sigma_1))\cap \cdots \cap \eig(A(\sigma_q)), 
$$ 
which violates the necessary condition given in~\eqref{eq:intersectionofeigensets}. Thus, we conclude that if there exists a uniformly controllable pair $(A, B)$ in $\V(G)$, then the subgraph $H$ has to admit a Hamiltonian decomposition.   
\end{proof}

\subsection{Minimality of Condition-B}\label{ssec:minimalB}
In this subsection, we show that the digraphs $G\in\mathcal{G}$ that are weakly connected and minimal with respect to condition-A are the ones satisfying condition-B. To proceed, we introduce, for each $n \ge 0$, a set of digraphs as follows: 
\begin{equation*}\label{eq:defkn}
\mathcal{K}_{n}:= \{G\in \mathcal{G}_{n,1} \mid \mbox{$G$ satisfies condition-B}\}.
\end{equation*}
We then let 
\begin{equation*}\label{eq:defk}
\mathcal{K}: = \cup^\infty_{n = 0} \mathcal{K}_n.
\end{equation*} 
For clarity of presentation, we will now use letter $K$ to denote a digraph in $\mathcal{K}$ for the remainder of the section. We establish below the following result: 

\begin{proposition}\label{prop:redminimal}
	A digraph $G\in\mathcal{G}$ is minimal with respect to condition-A  if and only if it is a disjoint union of $K_i$ where each $K_i$ belongs to $\mathcal{K}$.  
\end{proposition}

We illustrate in Fig.~\ref{fig:condensation} edge-reductions of the digraph in Fig.~\ref{fig:sparsepair} into disjoint unions of digraphs in $\K$.  

\begin{figure}[ht]
\begin{center}
\includegraphics[width = 0.6\textwidth]{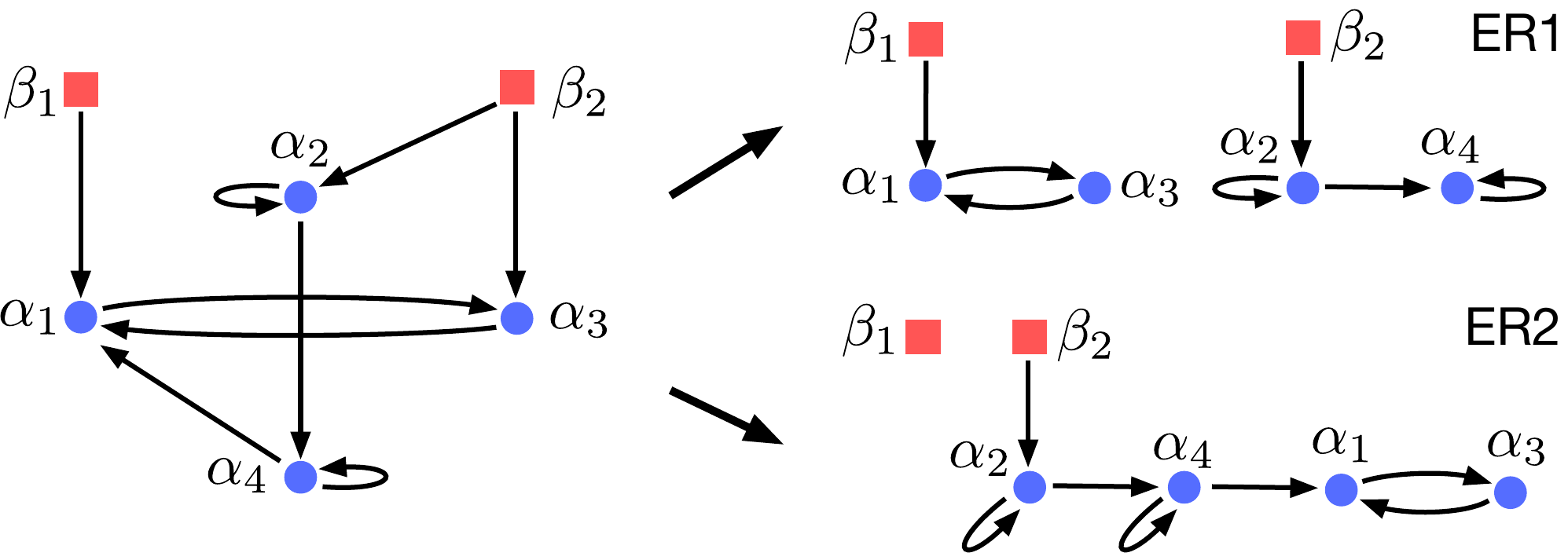}
\caption{\small The digraph on the left is from Fig.~\ref{fig:sparsepair}. It satisfies condition-A. We give two different edge-reductions (ERs) of the digraph and obtain disjoin unions of digraphs in $\mathcal{K}$ on the right. For ER1, we remove edges $\beta_2\alpha_3$ and $\alpha_4\alpha_1$ from the left. After the reduction, the two disjoint digraphs belong to $\K_2$. 
For ER2, we remove edges $\beta_1\alpha_1$ and $\beta_2\alpha_3$. After the reduction, the two disjoint digraphs belong to $\K_0$ and $\K_4$, respectively. 
}\label{fig:condensation}
\end{center}
\end{figure}

Prop.~\ref{prop:redminimal} will be established after a sequence of lemmas. We will first show that the digraphs in $\K$ satisfy condition-A and, next, show that these digraphs are minimal with respect to condition-A. They are done in Lemmas~\ref{lem:KsubsetG} and~\ref{lem:googl}, respectively.  
After that, we show that every digraph $G\in \mathcal{G}$, minimal with respect to condition-A, is a disjoint union of the digraphs in~$\mathcal{K}$. This is done in Lemma~\ref{lem:ba}.   

We start with the following lemma:

\begin{lemma}\label{lem:KsubsetG}
Every digraph $K\in \mathcal{K}$ satisfies condition-A.   	
\end{lemma}

\begin{proof}
Let $S$ be the skeleton digraph of $K$. Then, $S$ is rooted by item B1 of Theorem~\ref{thm:minimal}. It follows that $K$ is rooted with the control-node $\beta$ being the single root. 
In particular, $K$ is accessible to the root~$\beta$.   
Next, we let $H$ be the subgraph of $K$ induced by the state-nodes. We need to show that $H$ admits a Hamiltonian decomposition. But, this follows from item B2 of Theorem~\ref{thm:minimal}. To see this, let $K_0, K_1,\ldots, K_N$ be the subgraphs of $K$ obtained from the strong component decomposition, with $K_0$ being the singleton $\{\beta\}$. Note that all the other $K_i$, for $i = 1,\ldots, N$, are cycles. Moreover, they are subgraphs of $H$ and form a Hamiltonian decomposition of~$H$.  
\end{proof}

We next have the following fact:

\begin{lemma}\label{lem:googl}
Every digraph $K\in \mathcal{K}$ is minimal 	with respect to condition-A.
\end{lemma}

\begin{proof}
We show that removal of any edge out of $K$ violates condition-A. We again let $K_0,\ldots, K_N$ be the strong components of $K$ obtained from the strong components decomposition, with $K_0$ being the singleton of control-node $\beta$. We now remove an edge out of $K$. There are two cases for the edge: (1) it belongs to a certain component $K_i$ for some $i \in \{1,\ldots, N\}$, or,  (2) it connects two different components.   

We first deal with case (1). 
Note that by item 2 of Def.~\ref{def:strongcompdecomp}, each strong component $K_i$ cannot be contained in any strongly connected subgraph of $K$ other than itself. 
Thus, the subgraph $H$ induced by the state-nodes admits a Hamiltonian decomposition if and only if each $K_i$ admits a Hamiltonian decomposition.   
Since every $K_i$, for $i = 1,\ldots, N$, is a cycle by item~B2, the Hamiltonian decomposition of $H$ is unique, given by the union of these $K_i$. 
It follows that removing an edge out of one of these cycles violates item A2.   

We now deal with case (2). 
Let $S$ be the skeleton digraph of $K$ and $w_0,\ldots, w_N$ be the nodes of $S$. By item~B1, the skeleton digraph $S$ is an arborescence. Thus, if we remove an edge $w_iw_j$ out of $S$, then $S$ is disconnected. Correspondingly, if we remove all the edges in the set $[w_iw_j]$ (defined in~\eqref{eq:squarewiwj}) out of $K$, then $K$ will be disconnected and, hence, is not accessible anymore. 
Finally, note that by the same item B1, $[w_iw_j]$ contains only one single edge, so removing the edge out of $K$ will violate item A1.     
\end{proof}

To establish Prop.~\ref{prop:redminimal},  it remains to prove the following result: 

\begin{lemma}\label{lem:ba}
	Let $G = (V, E)\in\mathcal{G}$ satisfy condition-A. Then, there exist subgraphs $K_i = (V_i, E_i)$, for $i = 1,\ldots, m$, of $G$ such that every $K_i$ belongs to $\mathcal{K}$ and 
	$V = \sqcup^m_{i =1} V_i$.
\end{lemma}

\begin{proof}
We first consider the special case where $G$ has a single control-node $\beta$. In this case, $G$ is rooted with $\beta$ the root. 
We show below that $G$ can be reduced to a digraph $K$ in $\mathcal{K}$. 

Let $H$ be the subgraph of $G$ induced by the state-nodes and $H_1,\ldots, H_N$ be a Hamiltonian decomposition of $H$ (so every $H_i$ is a cycle). For convenience, let $H_0:=\{\beta\}$ be the singleton of the control-node. Similar to the strong component decomposition, we build a digraph $S$ by condensing all the $H_i$ to single nodes $w_i$, for $i = 0,\ldots,N$, and by adding edges $w_iw_j$, for $i\neq j$, if there exists at least one edge from $H_i$ to $H_j$. With slight abuse of notation, we will still let $[w_iw_j]$ be the set of edges $v_iv_j$ in $G$ with $v_i$ belonging to $H_i$ and $v_j$ belonging to $H_j$. 

Since $G$ is rooted, the resulting digraph $S$ is also rooted with $w_0$ the unique root. Let $S'$ be a directed spanning tree of~$S$. Given $S'$, we remove edges out of $G$ as follows: If $w_i w_j$ is an edge of $S'$ and if $[w_iw_j]$ has more than one edge, then we keep one edge in the set and remove the others from $G$. If $w_iw_j$ is not an edge of $S'$, then we remove all the edges in $[w_iw_j]$ from $G$. 
We let $K$ be the trimmed subgraph of $G$. Then, it should be clear that $S'$ is the skeleton digraph of $K$. By construction, the digraph $K$ satisfies condition-B. 

We now consider the general case where $G$ has $m$ control-nodes $\beta_1,\ldots,\beta_m$ for $m \ge 1$.   
For each $\beta_i$, we let $V^*_{i}$ be the union of the control-node $\beta_i$ and the set of state-nodes accessible to~$\beta_i$. Since $G$ is accessible, the union of $V^*_{i}$ is the entire node set $V$ of $G$. We next let
$$
V_{i}:= V^*_{i} \backslash \cup^{i-1}_{k = 1} V^*_{k}, \quad \forall i = 1,\ldots, m.
$$
These $V_{i}$ then form a disjoint union of $V$. Note that $V_{i}$ is never empty because it always contains $\beta_i$. However, $V_{i}$ may not contain any state-node. For each $i = 1,\ldots, m$, we let $G_{i}$ be the subgraph of $G$ induced by $V_{i}$. There is only one single control-node, namely $\beta_i$, in~$G_{i}$. 

We show below that every $G_{i}$ satisfies condition-A. Note that if this is the case, then one can apply the edge-reduction to every $G_i$ to obtain a digraph $K_i\in \mathcal{K}$ as was described earlier in the proof.

We first show that every $G_{i}$ is rooted (and, hence, satisfies item A1). Specifically, we show that for any $\alpha_j\in V_{i}$, there is a path from $\beta_i$ to $\alpha_j$ within $G_{i}$. By construction, nodes in $V^*_{i}$ are accessible to $\beta_i$ and $V_{i}$ is a subset of $V^*_{i}$. Thus, there exists a path from $\beta_i$ to $\alpha_j$ in $G$. Clearly, every node along the path belongs to $V^*_{i}$. We show below that all of these nodes (on the path) belong to $V_{i}$, i.e., none of them belongs to $V^*_{k}$ for any $k = 1,\ldots, i-1$. 
This holds because otherwise, the endpoint $\alpha_j$ of the path will be accessible to $V^*_{k}$ for some $k \in \{1,\ldots, i-1\}$, which contradicts the fact that $\alpha_j\in V_{i}$.   

We next show that every $G_{i}$ satisfies item A2. Specifically, we need to show that the subgraph $H_{i}$ of $G_{i}$ induced by the state-nodes admits a Hamiltonian decomposition. Since $G$ satisfies item A2, there are disjoint cycles $H_1,\ldots, H_N$ that cover all the state-nodes. The key observation is that if a node $\alpha_j$ of $G_{i}$ belongs to a cycle $H_\ell$ some $\ell \in \{1,\ldots, N\}$, then all the nodes of the cycle belong to $G_{i}$. To see this, note that if $\alpha_j$ is accessible to $\beta_i$, then so is every node in the cycle $H_\ell$. Conversely, if $\alpha_j$ is not accessible to $\beta_k$, for $k = 1,\ldots,i-1$, then neither is any node in $H_\ell$. The above arguments then imply that all the state-nodes of $G_{i}$ are covered by a certain selection of disjoint cycles $H_{i_1},\ldots, H_{i_{N'}}$. These cycles then form a Hamiltonian decomposition of $H_{i}$.
\end{proof}

Prop.~\ref{prop:redminimal} is now established by Lemmas~\ref{lem:KsubsetG}, \ref{lem:googl}, and~\ref{lem:ba}. 

For the remainder of the section, we will focus only on the digraphs in $\mathcal{K}$.  In particular, we will establish the sufficiency of condition-A by showing that the digraphs in $\K$ are structurally controllable.  We can do this because the  digraphs in $\mathcal{K}$ are minimal with respect to condition-A and, moreover, the property of being structural controllability is monotone with respect to edge-set inclusion:

 \begin{lemma}\label{lem:monotone}
 	Let $G = (V, E)\in\mathcal{G}$ and $G' = (V, E')$ be a subgraph of $G$, with the same node set $V$ and $E'\subseteq E$. If $G'$ is structurally controllable, then so is $G$.
 \end{lemma}

\begin{proof}
The result directly follows from the fact that if $(A,B)$ is a uniformly controllable pair and is compliant with $G'$, then it is also compliant with $G$.
\end{proof}

\subsection{Sparsity Patterns in Matrix Form}\label{ssec:matrixform}
Let $K\in \K$ and $(A, b)$ be a pair in $\mathbb{V}(K)$. We use little $b$ to indicate the fact that $b$ is a column vector (since $m = 1$). The goal of the subsection is to introduce a matrix form to represent the sparsity pattern of $(A, b)$. This matrix form will be used later for explicit construction of a uniformly controllable pair.   
 
Let $\beta$ be the unique control-node in~$K$ and $K_i = (V_i, E_i)$, for $i = 1,\ldots, N$, be the cycles of $K$. These cycles form a Hamiltonian decomposition of the subgraph induced by the state-nodes. 
  Let $S$ be the skeleton digraph of $K$ and $w_0$ be its root. Let $w_1,\ldots, w_N$ be the nodes corresponding to the cycles $K_1,\ldots, K_N$.

Recall that $S$ is an arborescence and the depth of a node $w_i$ in $S$ is the  length of the unique path from the root $w_0$ to $w_i$. 
The depth of $w_0$ is $0$ by default. 
By relabelling the nodes $w_1,\ldots, w_N$ (and, hence, the cycles $K_1,\ldots, K_N$), if necessary, we can assume that 
\begin{equation}\label{eq:monotoneinc}
	1 = \dep(w_1)  \le \cdots \le \dep(w_N).
\end{equation}

Let $n_i := |V_i|$ and we have that $\sum^N_{i = 1} n_i = n$.  
Without loss of generality, we assume that the first $n_1$ nodes of $K$ belong to $K_1$, the next $n_2$ nodes belong to $K_2$ and, in general,   
\begin{equation}\label{eq:nodeset}
V_i = \left \{\alpha_{s_{i - 1} + 1}, \ldots, \alpha_{s_{i-1} + n_i} \right \}, \quad \forall i = 1,\ldots, N,
\end{equation}
where $s_0:=0$ and $s_i := \sum^{i}_{k = 1} n_k$ for $i \geq 1$.    
Moreover, by relabeling the nodes within each $K_i$, if necessary, we can assume that the edge set $E_i$ of $K_i$ is given by
\begin{equation}\label{eq:edgeset}
E_i = \left \{\alpha_{s_{i - 1} + 1}\alpha_{s_{i - 1} + 2}, \alpha_{s_{i - 1} + 2}\alpha_{s_{i - 1} + 3}, \ldots, \alpha_{s_{i - 1} + n_i}\alpha_{s_{i - 1} + 1} \right \}.
\end{equation}

We now return to the sparse pair $(A,b)$. By the way we label the state-nodes of $K$, we have the following fact:

\begin{lemma}\label{lem:matrixform}
	The pair $(A, b)\in \mathbb{V}(K)$ satisfies the following conditions: 
\begin{enumerate}
\item The matrix $A$ is lower block triangular:  
$$
A = 
\begin{bmatrix}
	A_{11} & 0 & \cdots & 0 \\
	A_{21} & A_{22} & \cdots & 0 \\
	\vdots & \vdots & \ddots & \vdots \\
	A_{N1} & A_{N2} & \cdots & A_{NN}
\end{bmatrix}, 
$$
where each block $A_{ij}$ is $n_i\times n_j$.  
\item Every diagonal block $A_{ii}$ takes the following form:
\begin{equation}\label{eq:defAii}
A_{ii}= 
\begin{bmatrix}
	0 & 0 & \cdots & 0 & a_{i, 1n_i} \\
	a_{i,21} & 0 & \cdots & 0 & 0 \\
	0 & a_{i,32} & \cdots & 0 & 0 \\
	\vdots & \vdots & \ddots & \vdots & \vdots \\
	0 & 0 & \cdots & a_{i,n_in_i - 1} & 0
\end{bmatrix}. 
\end{equation}
\item Partition the vector $b = [b_1;\ldots; b_N]$, where each $b_i$ is $n_i$-dimensional. 
For each $i = 1,\ldots, N$, there is at most one nonzero block or vector among $\{A_{i1}, \ldots, A_{i,i-1}, b_i\}$. Moreover, the nonzero block or vector has only one nonzero entry.  
\end{enumerate}
\end{lemma}

\begin{proof}
	The lower block triangular structure of $A$ follows from~\eqref{eq:monotoneinc} and~\eqref{eq:nodeset}. Specifically, if $w_j w_i$ is an edge, then $\dep(w_i) > \dep(w_j)$. Thus, the increasing sequence~\eqref{eq:monotoneinc} implies that $j < i$ and, hence, the corresponding $A_{ij}$ is below the diagonal. 
	The second item of the lemma follows from the fact that every $K_i$ is a cycle and the way of labeling the nodes in $K_i$ as given in~\eqref{eq:edgeset}. The third item follows from the fact that the skeleton graph $S$ is an arborescence and $[w_iw_j]$ is a singleton for every edge $w_iw_j$ of $S$. 
\end{proof}

\subsection{Sufficiency of Condition-B}\label{ssec:sufficiency}
In this subsection, we establish the following result: 

\begin{proposition}\label{prop:sufficiency}
	Every minimal digraph $K\in \mathcal{K}$ is structurally controllable.   
\end{proposition}
 
To establish Prop.~\ref{prop:sufficiency}, we construct below a uniformly controllable pair $(A, b)\in \V(K)$. 
It takes two steps: We will  start by finding a pair $(A(0), b(0))$ such that the corresponding finite-dimensional system is controllable and, then, extend $(A(0), b(0))$ to a pair of functions $(A, b)$ over the entire interval $\Sigma = [0,1]$ so that $(A, b)$ is uniformly controllable.

With a slight abuse of terminology, we say that a pair $(A(0), b(0))\in \R^{n\times n}\times \R^n$ is {\em compliant} with the digraph~$K$ if each nonzero entry of $A(0)$ and $b(0)$ corresponds to an edge of~$K$.  
We assume that $(A(0), b(0))$ takes the form given in Lemma~\ref{lem:matrixform}. 
Let $A_{ii}(0)$ be the $ii$th block of $A(0)$, and the dimension of $A_{ii}(0)$ is $n_i\times n_i$.   
%The expression inside the absolute value is simply the product of the $\star$-entires of $A_{ii}(0)$. 
We have the following fact: 

\begin{lemma}\label{lem:finitesystem}
There exists a pair $(A(0), b(0))\in \R^{n\times n}\times \R^n$, compliant with $K$, such that the following two items hold: 
\begin{enumerate}
\item For any $i = 1,\ldots, N$, the $\star$-entries of $A_{ii}(0)$ are positive. Furthermore, let 
\begin{equation}\label{eq:defmu}
r_{i}(0) := \big ( a_{i,1n_i}(0)a_{i,21}(0)\cdots a_{i,n_in_i-1}(0) \big )^{\frac{1}{n_i}}.  
\end{equation}
 Then, $0 < r_1(0) < \cdots < r_N(0)$. 
\item The controllability matrix $C(A(0),b(0))$ is nonsingular.   
\end{enumerate}
\end{lemma}

\begin{proof}
By Remark~\ref{rmk:comparison}, condition-A is sufficient for~$K$ to be structural controllable for finite-dimensional linear systems. Thus, there exist controllable pairs $(A(0), b(0))$ compliant with $K$. Moreover, these controllable pairs are open and dense in $\R^{n\times n}\times \R^n$ with respect to the standard Euclidean topology~\cite{lin1974structural}. 
Now, let $(A(0), b(0))$ be chosen such that item~1 of the lemma is satisfied. If $C(A(0),b(0))$ is nonsingular, then the proof is complete. Otherwise, we can perturb $(A(0), b(0))$ to obtain a controllable pair $(A'(0), b'(0))$, arbitrarily close to $(A(0),b(0))$.
By~\eqref{eq:defmu}, each $r_i(0)$ is continuous in the $\star$-entires of $A_{ii}(0)$, so the pair $(A'(0), b'(0))$ will still satisfy item~1 of the lemma as long as the perturbation is sufficiently small.    	
\end{proof}

We will now extend the pair $(A(0), b(0))$ described in Lemma~\ref{lem:finitesystem} to a pair of matrix-valued functions $(A, b)$ that belongs to $\mathbb{V}(K)$. First, we define a positive real number: 
\begin{equation}\label{eq:defkappa}
\kappa:= \frac{1}{2}\min \left \{ \frac{r_{i + 1}(0)}{r_{i}(0)} - 1 \mid i = 1,\ldots, N -1 \right \}.
\end{equation} 
By item~1 of Lemma~\ref{lem:finitesystem}, $\kappa$ is well defined and is positive. 
We next define a linear function $\rho: \Sigma\to \R$ as follows: 
\begin{equation}\label{eq:defrho}
\rho(\sigma): = \kappa \sigma + 1. 
\end{equation}
Because $\kappa$ is positive, $\rho$ is everywhere nonzero and strictly monotonically increasing.  
We then let
\begin{equation}\label{eq:abpair}
	A(\sigma): = \rho(\sigma) A(0) \quad \mbox{and}\quad b(\sigma):= b(0), \quad \forall \sigma\in \Sigma. 
\end{equation} 
Since $(A(0), b(0))$ is compliant with $K$, so is $(A(\sigma),b(\sigma))$ for all $\sigma\in \Sigma$.   
It follows that $(A, b)\in \mathbb{V}(K)$. 
To establish Prop.~\ref{prop:sufficiency}, it now remains to establish the following fact: 

\begin{lemma}\label{lem:sufficientcondition}
	The pair $(A, b)$ defined in~\eqref{eq:abpair}, with $\rho$ given in~\eqref{eq:defrho}, is uniformly controllable. 
\end{lemma}

\begin{proof}
The following condition, adapted from~\cite{helmke2014uniform}, is a sufficient condition for $(A,b)$ to be uniformly controllable:
\begin{enumerate}
	\item For every $\sigma \in \Sigma$, the finite-dimensional linear system $(A(\sigma), b(\sigma))$ is controllable;
	\item For every $\sigma\in \Sigma$, the eigenvalues of $A(\sigma)$ have algebraic multiplicity one; 
	\item Let $\eig(A(\sigma))$ be the set of eigenvalues of $A(\sigma)$. If $\sigma\neq \sigma'$, then $\eig(A(\sigma))\cap \eig(A(\sigma')) = \varnothing$. 
\end{enumerate} 	
To show that the above three items are satisfied for the given $(A, b)$, we need some preliminaries. 

To that end, we extend each $r_i(0)$, for $i = 1,\ldots, N$, defined in~\eqref{eq:defmu} to a scalar function $r_i:\Sigma\to \R$. For each $\sigma\in \Sigma$,  we let  
$r_i(\sigma)$ be defined in the same way as was in~\eqref{eq:defmu}, but with the argument $0$ replaced with~$\sigma$. Because $A(\sigma) = \rho(\sigma) A(0)$ with $\rho(\sigma)$ positive, we have that 
$r_i(\sigma) = \rho(\sigma) r_i(0)$.  
Also, since $\rho$ is everywhere nonzero and strictly monotonically increasing, so is every~$r_i$. 
Further, for any $i = 1,\ldots, N - 1$, we use the fact that $r_i(1) = \rho(1)r_i(0) = (\kappa + 1)r_i(0)$ to obtain that 
$$
r_{i + 1}(0) - r_{i}(1) =  r_i(0) 
\left(
\nicefrac{r_{i + 1}(0)}{r_{i}(0)} - 1 - \kappa\right ) > 0,  
$$
where the inequality follows from the construction of~$\kappa$ given in~\eqref{eq:defkappa}. 
This inequality, combined with the monotonicity of each $r_i$, imply that if $i\neq j$, then 
\begin{equation}\label{eq:disjoint}
r_i(\sigma) \neq r_j(\sigma'), \quad \forall \sigma, \sigma'\in \Sigma,
\end{equation}
i.e., the images of $r_i$ and $r_j$ do not overlap.

With the above preliminaries, we now return to the proof that the three items given at the beginning of the proof are satisfied for the pair $(A, b)$ defined in~\eqref{eq:abpair}.  
\vspace{.1cm}

\noindent
{\em Proof that item 1 is satisfied.} 
Because $A(\sigma)= \rho(\sigma) A(0)$, the two controllability matrices $C(A(\sigma), b(\sigma))$ and $C(A(0), b(0))$ (which are square matrices) satisfy the following relation: 
\begin{equation}\label{eq:diagonalll}
C(A(\sigma), b(\sigma)) =  C(A(0), b(0))\operatorname{diag} [1, \rho(\sigma),\cdots,  \rho^{n-1}(\sigma)].
\end{equation} 
By item 2 of Lemma~\ref{lem:finitesystem}, $C(A(0), b(0))$ is nonsingular. Since $\rho(\sigma)$ is positive for all $\sigma\in \Sigma$, the diagonal matrix next to $C(A(0), b(0))$ in~\eqref{eq:diagonalll} is also nonsingular. We thus conclude that $C(A(\sigma), b(\sigma))$ is nonsingular and, hence, the finite-dimensional linear system $(A(\sigma), b(\sigma))$ is controllable.
\vspace{.1cm}

\noindent
{\em Proof that item 2 is satisfied.}  
First, note that every matrix $A(\sigma)$ is lower block triangular. Thus, the eigenvalues of $A(\sigma)$ are the union of the eigenvalues of the diagonal blocks $A_{ii}(\sigma)$, for $i = 1,\ldots, N$. 
We next note that the sparsity pattern of $A_{ii}(\sigma)$ is given in~\eqref{eq:defAii}. In particular, the characteristic polynomial of $A_{ii}(\sigma)$ can be computed explicitly as follows: 
$$
\det(\lambda I - A_{ii}(\sigma)) = \lambda^{n_i} - r^{n_i}_i(\sigma), 
$$
where $r_i(\sigma) > 0$ is defined earlier in the proof. 
The roots of the above polynomial are given by:  
$$
\eig(A_{ii}(\sigma)) = 
\left \{ 
r_i(\sigma) e^{\frac{\mathrm{i}2\pi k}{n_i}}
\mid k = 0,\ldots, n_i - 1
\right \},
$$   
so the $n_i$ eigenvalues of $A_{ii}(\sigma)$ are pairwise distinct. 
Moreover, the $N$ sets $\eig(A_{ii}(\sigma))$, for $i = 1,\ldots, N$, are pairwise disjoint. This holds because if $\lambda_i \in \eig(A_{ii}(\sigma))$ and $\lambda_j\in A_{jj}(\sigma)$, with $i\neq j$, then 
$$|\lambda_i| = r_i(\sigma) \quad  \mbox{and} \quad |\lambda_j| = r_j(\sigma).$$
By~\eqref{eq:disjoint}, if $i\neq j$, then $r_i(\sigma) \neq r_j(\sigma)$ and, hence, $\lambda_i \neq \lambda_j$. Thus, the matrix $A(\sigma)$ has $n$ distinct eigenvalues for all $\sigma\in \Sigma$, i.e., the eigenvalues of $A(\sigma)$ have algebraic multiplicity one.   
\vspace{.1cm}

\noindent
{\em Proof that item 3 is satisfied.} 
Let $\sigma$ and $\sigma'$ be two distinct points in~$\Sigma$. Without loss of generality, we assume that $\sigma < \sigma'$. Let $\lambda$ and $\lambda'$ be two arbitrary eigenvalues of $A(\sigma)$ and $A(\sigma')$, respectively. We show below that $\lambda \neq \lambda'$. Since $A(\sigma)$ and $A(\sigma')$ are lower block triangular, $\lambda$ and $\lambda'$ are eigenvalues of certain diagonal blocks of $A(\sigma)$ and $A(\sigma')$, respectively.  
Without loss of generality, we assume that $\lambda\in \eig(A_{ii}(\sigma))$ and $\lambda'\in \eig(A_{jj}(\sigma'))$. There are two cases: If $i = j$, then $|\lambda| = r_i(\sigma)$ and $|\lambda'| = r_i(\sigma')$. Since $r_i$ is strictly monotonically increasing, $r_i(\sigma) < r_i(\sigma')$ and, hence, $\lambda \neq \lambda'$. If $i\neq j$, then $|\lambda| = r_i(\sigma)$ and $|\lambda'| = r_j(\sigma')$. By~\eqref{eq:disjoint}, $r_i(\sigma) \neq r_j(\sigma')$. In either case, we have that $\lambda\neq \lambda'$. This completes the proof. 
\end{proof}

We further provide an example that illustrates the procedure for constructing a pair $(A, b)$ compliant with a given $K\in \cal{K}$.

\begin{example}{\em
	Consider the minimal graph $K\in \cal{K}$ in Fig.~\ref{fig:construction}.	
	\begin{figure}[ht]
\begin{center}
\includegraphics[width = 0.45\textwidth]{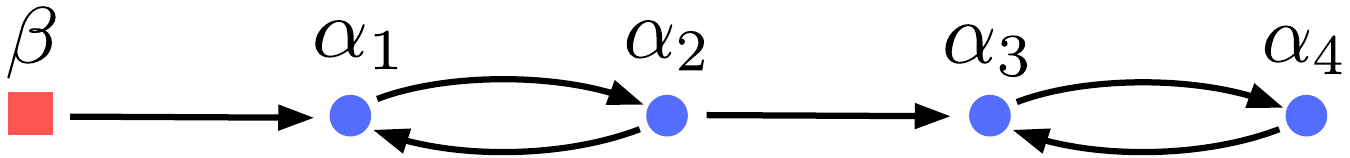}
\caption{\small  A minimal graph $K\in \cal{K}$. The two $2$-cycles in the subgraph~$H$ induced by the state-nodes form a Hamiltonian decomposition of~$H$.  
}\label{fig:construction}
\end{center}
\end{figure}

	To construct a controllable pair $(A, b)$ compliant with $K$, we first specify their values at $\sigma = 0$:
	$$
	A(0) = 
	\begin{bmatrix}
		0 & 1 & 0 & 0 \\
		1 & 0 & 0 & 0 \\
		0 & 1 & 0 & 3 \\
		0 & 0 & 3 & 0 
	\end{bmatrix}
	 \quad \mbox{and} \quad
	b(0) = 
	\begin{bmatrix}
		1 \\
		0 \\
		0 \\
		0
	\end{bmatrix}.
	$$
	Note that $(A(0), b(0))$ satisfies the two items of Lemma~\ref{lem:finitesystem}. Indeed, the controllability matrix computed below:
	$$
C(A(0), b(0)) = 
\begin{bmatrix}
	1 & 0 & 1 & 0\\
	0 & 1 & 0 & 1\\
	0 & 0 & 1 & 0 \\
	0 & 0 & 0 & 3
\end{bmatrix},
$$
is nonsingular. Also, the two scalars $r_1(0) = (a_{12}(0)a_{21}(0))^{\frac{1}{2}}$ and $r_2(0) = (a_{34}(0)a_{43}(0))^{\frac{1}{2}}$ associated with the two $2$-cycles are given by $r_1(0) = 1$ and $r_2(0) = 3$, so $r_1(0) < r_2(0)$.

Correspondingly, the scalar $\kappa$ introduced in~\eqref{eq:defkappa} takes value~$1$, so the linear function $\rho(\sigma)$ is simply given by $\rho(\sigma) = \sigma  + 1$. 
We then follow~\eqref{eq:abpair} to define $
A(\sigma):= (\sigma + 1) A(0)$ and $b(\sigma) := b(0), 
$ for all $\sigma \in \Sigma = [0,1]$. 

By construction, $(A(\sigma),b(\sigma))$ is a controllable pair for all $\sigma\in \Sigma$. The eigenvalues of $A(\sigma)$ are given by $\pm (1 + \sigma)$ and $\pm 3(1 + \sigma)$, all of which have algebraic multiplicity one. Furthermore, for any two distinct $\sigma$ and $\sigma'$ in $[0,1]$, $\eig(A(\sigma))$ does not intersect with $\eig(A(\sigma'))$. 
Thus, the three items given at the beginning of the proof of Lemma~\ref{lem:sufficientcondition} are satisfied, and we conclude that the pair $(A, b)$ is uniformly controllable.      
}
\end{example}

We now combine results established in this section and prove the two theorems formulated in Section~\ref{sec:formulationandresult}:  

\begin{proof}[Proof of Theorems~\ref{thm:main1} and~\ref{thm:minimal}]
The necessity of condition-A for structural controllability is established in Prop.~\ref{prop:necessary}. The proof of sufficiency relies on the use of the monotonicity property (Lemma~\ref{lem:monotone}): We have shown in Prop.~\ref{prop:redminimal} that condition-B is minimal with respect to condition-A  and, then, in Prop.~\ref{prop:sufficiency} that condition-B itself is sufficient for structural controllability.  
\end{proof}

\section{Conclusions and Outlooks}\label{sec:conclusions}
We have introduced and solved the structural controllability problem for linear ensemble systems over (finite unions of) closed intervals in $\R$. A necessary and sufficient condition is provided in Theorem~\ref{thm:main1} for a sparsity pattern to be structural controllable. The minimal sparsity patterns are further characterized in Theorem~\ref{thm:minimal}. 

Recall that in the definition of structural controllability (Def.~\ref{def:StrucContr}), we only need $(A, B)$ to be continuous. The condition can be made stronger by requiring that $(A, B)$ be $k$th continuously differentiable, for $k = 0,\ldots, \infty$, or even real-analytic, i.e., $k = \omega$. But, changing the condition does not affect the results. This holds because the $(A, b)$ pair constructed in Subsection~\ref{ssec:sufficiency} is, in fact, linear in $\sigma$. 

A relevant question we will aim to investigate in the future is formulated below: 
Let $G\in \mathcal{G}_{n,m}$ be structurally controllable digraph. 
For two nonnegative integers $k$ and $\ell$, we let 
$$
\V^{k,\ell}(G):= \big ( \mathrm{C}^k(\Sigma, \R^{n\times n}) \times \mathrm{C}^\ell(\Sigma, \R^{n\times m}) \big ) \cap \V(G),
$$
i.e., a pair $(A, B)$ belongs to $\V^{k,\ell}(G)$ if and only if $A$  (resp. $B$) is $k$th (resp. $\ell$th) continuously differentiable and $(A, B)$ is compliant with $G$. We endow the space $\mathrm{C}^k(\Sigma, \R^{n\times n})$ (resp. $\mathrm{C}^\ell(\Sigma, \R^{n\times m})$) with, e.g., the Whitney $\mathrm{C}^k$- (resp. $\mathrm{C}^\ell$-) topology. The subspace $\V^{k,\ell}(G)$ is then endowed with the subspace topology. Define a subset of $\V^{k,\ell}(G)$ as follows:
$$
\V^{k,\ell}_*(G):=\big \{(A, B)\in \V^{k,\ell}(G) \mid  
(A, B) \mbox{ is uniformly controllable}
\big \}.
$$
The subset $\V_*^{k,\ell}(G)$ can hardly be dense in $\V^{k,\ell}(G)$. We are interested in its openness.  More precisely, we ask when does $\V_*^{k,\ell}(G)$ contain open sets in $\V^{k,\ell}(G)$?  
On one hand, we conjecture that if the topology on the space of $A$-matrices is too coarse (e.g., the Whitney $\mathrm{C}^0$-topology), then $\V^{0,\ell}_*(G)$ does {\em not} contain an open set. The conjecture is based on the analysis carried out in Section~\ref{ssec:sufficiency}; specifically, the monotonicity of each (continuous) branch of eigenvalues of $A$ is crucial to uniform controllability of $(A,b)$. 
However, a perturbation of $A$, in the $\mathrm{C}^0$-sense, can easily violate such a property. 
On the other hand, we conjecture that if $k\geq 1$, then
 $\V_*^{k,\ell}(G)$ contains open sets in $\V^{k,\ell}(G)$ (the space of $A$-matrices is endowed with the Whitney $\mathrm{C}^k$-topology). 

Other potential extensions of the current work include characterizations of structurally controllable digraphs for linear ensemble systems whose parameterization spaces are circles and, further, for nonlinear ensemble systems (a prototype of sparse bilinear ensemble system is investigated in~\cite{chen2019controllability}).

\bibliographystyle{IEEEtran}    
\bibliography{struc.bib}

% Generated by IEEEtran.bst, version: 1.14 (2015/08/26)
\begin{thebibliography}{10}
\providecommand{\url}[1]{#1}
\csname url@samestyle\endcsname
\providecommand{\newblock}{\relax}
\providecommand{\bibinfo}[2]{#2}
\providecommand{\BIBentrySTDinterwordspacing}{\spaceskip=0pt\relax}
\providecommand{\BIBentryALTinterwordstretchfactor}{4}
\providecommand{\BIBentryALTinterwordspacing}{\spaceskip=\fontdimen2\font plus
\BIBentryALTinterwordstretchfactor\fontdimen3\font minus
  \fontdimen4\font\relax}
\providecommand{\BIBforeignlanguage}[2]{{%
\expandafter\ifx\csname l@#1\endcsname\relax
\typeout{** WARNING: IEEEtran.bst: No hyphenation pattern has been}%
\typeout{** loaded for the language `#1'. Using the pattern for}%
\typeout{** the default language instead.}%
\else
\language=\csname l@#1\endcsname
\fi
#2}}
\providecommand{\BIBdecl}{\relax}
\BIBdecl

\bibitem{milo2002network}
R.~Milo, S.~Shen-Orr, S.~Itzkovitz, N.~Kashtan, D.~Chklovskii, and U.~Alon,
  ``Network motifs: {S}imple building blocks of complex networks,''
  \emph{Science}, vol. 298, no. 5594, pp. 824--827, 2002.

\bibitem{baillieul2003information}
J.~Baillieul and A.~Suri, ``Information patterns and hedging {B}rockett's
  theorem in controlling vehicle formations,'' in \emph{42nd IEEE Conference on
  Decision and Control (CDC)}, vol.~1.\hskip 1em plus 0.5em minus 0.4em\relax
  IEEE, 2003, pp. 556--563.

\bibitem{rahmani2009controllability}
A.~Rahmani, M.~Ji, M.~Mesbahi, and M.~Egerstedt, ``Controllability of
  multi-agent systems from a graph-theoretic perspective,'' \emph{SIAM Journal
  on Control and Optimization}, vol.~48, no.~1, pp. 162--186, 2009.

\bibitem{pasqualetti2014controllability}
F.~Pasqualetti, S.~Zampieri, and F.~Bullo, ``Controllability metrics,
  limitations and algorithms for complex networks,'' \emph{IEEE Transactions on
  Control of Network Systems}, vol.~1, no.~1, pp. 40--52, 2014.

\bibitem{chen2019controlling}
X.~Chen, M.-A. Belabbas, and T.~Ba{\c s}ar, ``Controlling and stabilizing a
  rigid formation using a few agents,'' \emph{SIAM Journal on Control and
  Optimization}, vol.~57, no.~1, pp. 104--128, 2019.

\bibitem{dirr2021uniform}
G.~Dirr and M.~Sch{\"o}nlein, ``Uniform and $\mathrm{L}^q$-ensemble
  reachability of parameter-dependent linear systems,'' \emph{Journal of
  Differential Equations}, vol. 283, pp. 216--262, 2021.

\bibitem{chen2020controllability}
X.~Chen, ``Controllability issues of linear ensemble systems over
  multi-dimensional parameterization spaces,'' \emph{arXiv:2003.04529}, 2020.

\bibitem{lin1974structural}
C.-T. Lin, ``Structural controllability,'' \emph{IEEE Transactions on Automatic
  Control}, vol.~19, no.~3, pp. 201--208, 1974.

\bibitem{shields1976structural}
R.~Shields and J.~Pearson, ``Structural controllability of multiinput linear
  systems,'' \emph{IEEE Transactions on Automatic control}, vol.~21, no.~2, pp.
  203--212, 1976.

\bibitem{glover1976characterization}
K.~Glover and L.~Silverman, ``Characterization of structural controllability,''
  \emph{IEEE Transactions on Automatic control}, vol.~21, no.~4, pp. 534--537,
  1976.

\bibitem{mayeda1979strong}
H.~Mayeda and T.~Yamada, ``Strong structural controllability,'' \emph{SIAM
  Journal on Control and Optimization}, vol.~17, no.~1, pp. 123--138, 1979.

\bibitem{chapman2013strong}
A.~Chapman and M.~Mesbahi, ``On strong structural controllability of networked
  systems: {A} constrained matching approach,'' in \emph{2013 American Control
  Conference}.\hskip 1em plus 0.5em minus 0.4em\relax IEEE, 2013, pp.
  6126--6131.

\bibitem{jia2020unifying}
J.~Jia, H.~J. van Waarde, H.~L. Trentelman, and M.~K. Camlibel, ``A unifying
  framework for strong structural controllability,'' \emph{IEEE Transactions on
  Automatic Control}, vol.~66, no.~1, pp. 391--398, 2020.

\bibitem{liu2011controllability}
Y.-Y. Liu, J.-J. Slotine, and A.-L. Barab{\'a}si, ``Controllability of complex
  networks,'' \emph{Nature}, vol. 473, no. 7346, pp. 167--173, 2011.

\bibitem{commault2015single}
C.~Commault and J.-M. Dion, ``The single-input minimal controllability problem
  for structured systems,'' \emph{Systems \& Control Letters}, vol.~80, pp.
  50--55, 2015.

\bibitem{olshevsky2014minimal}
A.~Olshevsky, ``Minimal controllability problems,'' \emph{IEEE Transactions on
  Control of Network Systems}, vol.~1, no.~3, pp. 249--258, 2014.

\bibitem{olshevsky2015minimum}
------, ``Minimum input selection for structural controllability,'' in
  \emph{2015 American Control Conference (ACC)}.\hskip 1em plus 0.5em minus
  0.4em\relax IEEE, 2015, pp. 2218--2223.

\bibitem{sundaram2012structural}
S.~Sundaram and C.~N. Hadjicostis, ``Structural controllability and
  observability of linear systems over finite fields with applications to
  multi-agent systems,'' \emph{IEEE Transactions on Automatic Control},
  vol.~58, no.~1, pp. 60--73, 2012.

\bibitem{tsopelakos2018classification}
A.~Tsopelakos, M.-A. Belabbas, and B.~Gharesifard, ``Classification of the
  structurally controllable zero-patterns for driftless bilinear control
  systems,'' \emph{IEEE Transactions on Control of Network Systems}, vol.~6,
  no.~1, pp. 429--439, 2018.

\bibitem{triggiani1975controllability}
R.~Triggiani, ``Controllability and observability in {B}anach space with
  bounded operators,'' \emph{SIAM Journal on Control}, vol.~13, no.~2, pp.
  462--491, 1975.

\bibitem{li2007ensemble}
J.-S. Li and N.~Khaneja, ``Ensemble control of linear systems,'' in
  \emph{Decision and Control (CDC), 46th IEEE Conference on}.\hskip 1em plus
  0.5em minus 0.4em\relax IEEE, 2007, pp. 3768--3773.

\bibitem{helmke2014uniform}
U.~Helmke and M.~Sch{\"o}nlein, ``Uniform ensemble controllability for
  one-parameter families of time-invariant linear systems,'' \emph{Systems \&
  Control Letters}, vol.~71, pp. 69--77, 2014.

\bibitem{li2015ensemble}
J.-S. Li and J.~Qi, ``Ensemble control of time-invariant linear systems with
  linear parameter variation,'' \emph{IEEE Transactions on Automatic Control},
  vol.~61, no.~10, pp. 2808--2820, 2015.

\bibitem{chen2015controllability}
X.~Chen, M.-A. Belabbas, and T.~Ba{\c{s}}ar, ``Controllability of formations
  over directed time-varying graphs,'' \emph{IEEE Transactions on Control of
  Network Systems}, vol.~4, no.~3, pp. 407--416, 2015.

\bibitem{belabbas2013sparse}
M.-A. Belabbas, ``Sparse stable systems,'' \emph{Systems \& Control Letters},
  vol.~62, no.~10, pp. 981--987, 2013.

\bibitem{chen2019controllability}
X.~Chen, ``Controllability of continuum ensemble of formation systems over
  directed graphs,'' \emph{Automatica}, vol. 108, p. 108497, 2019.

\bibitem{grasse2004vector}
K.~A. Grasse, ``A vector-bundle version of a theorem of {V}. {D}ole{\v{z}}al,''
  \emph{Linear Algebra and Its Applications}, vol. 392, pp. 45--59, 2004.

\end{thebibliography}

\section*{Appendix}
We provide here extensions of Theorem~\ref{thm:main1} by addressing 
 $\mathrm{L}^p$-controllability, for $1\leq p < \infty$. We will first review a few preliminary results, next introduce the notion of structural $\mathrm{L}^p$-controllability and, then, show that the same condition, condition-A given in Theorem~\ref{thm:main1}, is still necessary and sufficient for a digraph $G\in \mathcal{G}$ to be structurally $\mathrm{L}^p$-controllable.  

Let $\Sigma = [0,1]$ be the parameterization space. We represent a linear ensemble system over $\Sigma$ by a matrix pair $(A, B)$, with $A\in \mathrm{C}^0(\Sigma, \R^{n\times n})$ and  $B\in \mathrm{L}^p(\Sigma, \R^{n\times m})$, i.e., each entry of $B$ has finite $\mathrm{L}^p$-norm. Correspondingly, profiles $x_\Sigma(t)$ of the linear ensemble system are now elements in $\mathrm{L}^p(\Sigma, \R^n)$. 

The pair $(A, B)$ is said to be $\mathrm{L}^p$-controllable if for any initial profile $x_\Sigma(0)$, any target profile $\hat x_\Sigma$, and any error tolerance $\epsilon > 0$, there exist a time $T > 0$ and an integrable control input $u:[0,T]\to \R^m$ such that the solution $x_\Sigma(t)$  generated by the ensemble system $(A, B)$ satisfies $\|x_\Sigma(T) - \hat x_\Sigma\|_{\rm L^p} < \epsilon$. 

Similarly, one defines the $\mathrm{L}^p$-controllable subspace, denoted by $\mathcal{L}^p(A, B)$, as the $\mathrm{L}^p$-closure of the vector space spanned by the columns of $A^kB$, for $k\ge 0$. 
A counterpart of Lemma~\ref{lem:controllablesubspace} is given below (which is also adapted from~\cite{triggiani1975controllability}): 

\begin{lemma}\label{lem:controllablesubspace2}
A pair $(A, B)\in \mathrm{C}^0(\Sigma, \R^{n\times n})\times \mathrm{L}^p(\Sigma, \R^{n\times m})$ is $\mathrm{L}^p$-controllable if and only if $\mathcal{L}^p(A, B) = \mathrm{L}^p(\Sigma, \R^n)$. 
\end{lemma}

We further note the following fact as a counterpart of Lemma~\ref{lem:subensemble} (a proof can be found in~\cite{chen2020controllability}):

\begin{lemma}\label{lem:subensemble1}
Let $(A, B)\in \mathrm{C}^0(\Sigma, \R^{n\times n})\times \mathrm{L}^p(\Sigma, \R^{n\times m})$.  
	For a closed sub-interval $\Sigma'$ of $\Sigma$, let $A'$ and $B'$ be obtained by restricting $A$ and $B$ to $\Sigma'$. If $(A', B')$ is not $\mathrm{L}^p$-controllable, then neither is $(A, B)$.  
\end{lemma}

Given a digraph $G\in \mathcal{G}_{n,m}$,  
let $\V^p(G)$ be the set of matrix pairs $(A, B)\in \mathrm{C}^0(\Sigma, \R^{n\times n})\times \mathrm{L}^p(\Sigma, \R^{n\times m})$ compliant with~$G$. The following definition is a variation of Def.~\ref{def:StrucContr}: 

\begin{definition}
	A digraph $G\in \mathcal{G}$ is said to be \textbf{structurally $\mathrm{L}^p$-controllable} if there is an $\mathrm{L}^p$-controllable pair in $\mathbb{V}^p(G)$.  
\end{definition}

With the definition above, we have the following result: 

\begin{theorem}\label{thm:structurallpcontrollability}
	A digraph $G\in \mathcal{G}$ is structurally $\mathrm{L}^p$-controllable if and only if condition-A given in Theorem~\ref{thm:main1} is satisfied. 
\end{theorem}

\begin{proof}
The proof for sufficiency of condition-A is straightforward. If $G$ satisfies condition-A, then, by Theorem~\ref{thm:main1}, there is a continuous matrix pair $(A, B)$ such that $(A, B)$ is uniformly controllable and compliant with~$G$. By Lemma~\ref{lem:controllablesubspace}, $\mathcal{L}(A, B) = \mathrm{C}^0(\Sigma, \R^n)$. Since $\mathcal{L}(A, B) \subseteq \mathcal{L}^p(A, B)$ and since the $\mathrm{L}^p$-closure of $\mathrm{C}^0(\Sigma, \R^n)$ is $\mathrm{L}^p(\Sigma, \R^n)$, by Lemma~\ref{lem:controllablesubspace2}, $(A, B)$ is also $\mathrm{L}^p$-controllable.     

The necessity of condition-A is also not too hard to establish. The arguments for necessity of accessibility of $G$ to control-nodes are the same as the ones in the proof of Prop.~\ref{prop:necessary}. The proof of this part is thus omitted. 

It remains to show that the subgraph $H$ of $G$ induced by the state-nodes has to admit a Hamiltonian decomposition in order for $G$ to be structurally $\mathrm{L}^p$-controllable. 
Recall from the arguments in the proof of Prop.~\ref{prop:necessary} that if the subgraph $H$ does not admit any Hamiltonian decomposition, then for any $(A, B)\in \V^p(G)$, $\det A \equiv 0$. Thus, it suffices to show that any matrix pair $(A, B)$, with $\det A \equiv 0$, is not $\mathrm{L}^p$-controllable.

The proof will be carried out by induction on~$n$. For the base case where $n = 1$, $\det (A) \equiv 0$ implies that $A \equiv 0$. Because $\mathcal{L}^p(0,B)$ is the column space of $B$ which is finite dimensional, it follows from  Lemma~\ref{lem:controllablesubspace2} that $(0,B)$ is not $\mathrm{L}^p$-controllable.   

For the inductive step, we assume that the result holds for all $k \leq (n-1)$ and we prove for~$n$. 
Since $\operatorname{rank}A(\sigma)$ takes value from the finite set $\{0,\ldots, n\}$, there is a $\sigma^*\in \Sigma$ such that $\operatorname{rank}A(\sigma^*)$ achieves the maximal value, denoted by $k$, over $\Sigma$. 
Since $\det A \equiv 0$, $k$ is strictly less than $n$.  
 
Because $A$ is continuous in $\sigma$ and $\operatorname{rank} A$ is locally nondecreasing in $\sigma$, there is a closed interval $\Sigma':= [\sigma'_-, \sigma'_+]$ in $\Sigma$, with $\sigma'_-< \sigma^* <\sigma'_+$, such that $ \operatorname{rank} A(\sigma) = k$ for all $\sigma\in \Sigma'$. 
It is known~\cite{grasse2004vector} that 
there exists a continuous function $P: \Sigma'\to \GL(n,\R)$ such that 
$A':=PAP^{-1} = [0, A'_{12};0, A'_{22}]$, 
where $A'_{22}$ is $k\times k$. We next let $B':= P B$ and partition $B' = [B'_1; B'_2]$, where $B'_2$ is $k\times m$.  

Next, consider the linear ensemble system given by the $(A',B')$ pair constructed above: 
\begin{equation}\label{eq:secondstep}
	\begin{bmatrix}
		\dot x_1(t,\sigma) \\
		\dot x_2(t,\sigma)
	\end{bmatrix} = 
	\begin{bmatrix}
 0 & A'_{12}(\sigma) \\
 0 & A'_{22}(\sigma) 	
\end{bmatrix}
\begin{bmatrix}
		x_1(t,\sigma) \\
		x_2(t,\sigma)
	\end{bmatrix}  +  
	\begin{bmatrix} 
		B'_1(\sigma) \\
		B'_2(\sigma) 
	\end{bmatrix} u(t), \quad\forall \sigma\in \Sigma'.
\end{equation}
By construction, the above system is obtained by first restricting the $(A, B)$ pair to $\Sigma'$ and, then, applying a similarity transformation via~$P$.  
Note that similarity transformation preserves $\mathrm{L}^p$-controllability. 
Thus, by Lemma~\ref{lem:subensemble1}, to show that $(A, B)$ is not $\mathrm{L}^p$-controllable, it suffices to show that $(A',B')$ is not.  We now consider two cases:
\vspace{.1cm}

\noindent
{\em Case 1: $\det A'_{22}\equiv 0$.} Note that the dynamics of $x_2(t,\sigma)$ in~\eqref{eq:secondstep} do not depend on the dynamics of $x_1(t,\sigma)$: 
\begin{equation}\label{eq:substate}
\dot x_2(t,\sigma) = A'_{22}(\sigma) x_2(t,\sigma) + B'_2(\sigma) u(t), \quad \forall \sigma\in \Sigma'. 
\end{equation}  
It should be clear that if system~\eqref{eq:substate} is not $\mathrm{L}^p$-controllable, then neither is~\eqref{eq:secondstep}. Since $\det A'_{22}\equiv 0$ and $A'_{22}$ is $k\times k$ with $k < n$, we apply the induction hypothesis to conclude that system~\eqref{eq:substate} is not $\mathrm{L}^p$-controllable. The proof is then done.  
\vspace{.1cm}

\noindent
{\em Case 2: $\det A'_{22}\not\equiv 0$.} 
Because $A'_{22}$ is continuous, there is a closed interval $\Sigma'' := [\sigma''_-, \sigma''_+]$ in $\Sigma'$, with $\sigma''_- < \sigma''_+$, such that $\det A'_{22}(\sigma)\neq 0$ for all $\sigma\in \Sigma''$. It follows that $A'_{22}$ is invertible when restricted to $\Sigma''$. 
Let $P':\Sigma'' \to \GL(n,\R)$ be defined as follows: 
$$P':= 
\begin{bmatrix}
I_{n - k} & -A'_{12}A_{22}'^{-1} \\ 
0 & I_{k}
\end{bmatrix}.$$ 
The inverse $P'^{-1}$ is simply given by 
$$P'^{-1} = 
\begin{bmatrix}
I_{n - k} & A'_{12}A_{22}'^{-1} \\ 
0 & I_{k}
\end{bmatrix}.
$$ 
Next, define $A'':= P' A' P'^{-1}$ and $B'':= P' B'$. By computation, we have that 
$A'' = [0, 0; 0, A'_{22}]$. 
Correspondingly, we partition $B'' = [B''_1; B''_2]$, where $B''_2$ is $k\times m$.  

Consider the linear ensemble system given by the $(A'',B'')$ pair,   
which is obtained by first restricting $(A',B')$ from $\Sigma'$ to $\Sigma''$ and, then, applying a similarity transformation via $P'$. 
Thus, to show that $(A',B')$ is not $\mathrm{L}^p$-controllable, we only need to show that $(A'',B'')$ is not.  
For any $f\in \mathcal{L}^p(A'',B'')$, we decompose $f = [f_1; f_2]$ with $f_1$ of dimension $(n - k)$. Then, by the structure of matrix $A''$, we have that $f_1$ belongs to the column space of $B''_1$, which is finite dimensional. It follows that $\mathcal{L}^p(A'',B'')\neq \mathrm{L}^p(\Sigma'', \R^n)$ and, hence, by Lemma~\ref{lem:controllablesubspace2}, $(A'',B'')$ is not $\mathrm{L}^p$-controllable. This completes the proof.   	
\end{proof}

\end{document}